\documentclass[11pt]{article}
\usepackage{complexity}

\usepackage[english]{babel}
\usepackage[utf8x]{inputenc}
\usepackage[T1]{fontenc}
\usepackage{enumitem}
\usepackage[margin=1.in,marginparwidth=1.75cm]{geometry}

\usepackage[algo2e, ruled, vlined]{algorithm2e}
\usepackage{amsmath}
\usepackage{amsfonts}
\usepackage{graphicx}
\usepackage{amsthm}
\usepackage{amssymb}
\usepackage{dsfont}
\usepackage[algo2e, ruled, vlined]{algorithm2e}
\usepackage[noend]{algorithmic} 
\usepackage{bbm}

\usepackage[T1]{fontenc}
\usepackage{microtype}          %
\usepackage[full]{textcomp}

\usepackage{lipsum}
\usepackage{xspace}
\usepackage{relsize}
\usepackage{framed}
\usepackage{xcolor}
\usepackage{graphicx}
\usepackage{multirow}
\usepackage{subcaption}
\usepackage{float}

\usepackage[colorinlistoftodos]{todonotes}
\usepackage[colorlinks=true, allcolors=blue]{hyperref}

\newtheorem{theorem}{Theorem}

\newtheorem{lemma}{Lemma}

\newtheorem{remark}{Remark}

\theoremstyle{definition}
\newtheorem{definition}{Definition}

\newcommand{\calN}{{\mathcal{N}}}

\newcommand{\Zg}{\mathbb{Z}_{\geq 0}}
\newcommand{\Rg}{\mathbb{R}_{\geq 0}}
\newcommand{\Comp}{\#\text{Comps}}
\newcommand{\cN}{\mathcal{N}}

\newcommand{\mcp}{\mathcal{P}}
\newcommand{\pc}{\text{PC}}
\newcommand{\wpc}{\text{WPC}}

\newcommand{\knote}[1]{{\bf{\color{blue}[\tiny Karthik: #1]}}}
\newcommand{\cnote}[1]{{\bf{\color{magenta}[\tiny Chandra: #1]}}}

\usepackage{float}
\usepackage{booktabs}
\usepackage{multirow}
\usepackage{thmtools} 
\usepackage{thm-restate}
\usepackage{lipsum}

\newcommand{\Z}{\mathbb{Z}}

\newif\ifdraft
\draftfalse

\title{Hedgegraph Polymatroids\thanks{Univ. of Illinois, Urbana-Champaign, Urbana, IL 61801. Email: {\tt
      \{karthe, chekuri, weihaoz3\}@illinois.edu}. Supported in part by NSF grant CCF-2402667.}}
\author{Karthekeyan Chandrasekaran
\and Chandra Chekuri
\and Weihang Wang\thanks{Based on work done as a PhD student at University of Illinois, Urbana-Champaign.}
\and Weihao Zhu 
}
\date{}
\begin{document}

\maketitle

\pagenumbering{gobble}
\begin{abstract}

Graphs and hypergraphs combine expressive modeling power with algorithmic efficiency for a wide range of applications. 
Hedgegraphs generalize hypergraphs further by grouping hyperedges under a color/hedge. 
This allows hedgegraphs to model dependencies between hyperedges and leads to several applications. 
However, it poses algorithmic challenges. 
In particular, the cut function is not submodular, which has been a barrier to algorithms for connectivity. 
In this work, we introduce two alternative partition-based measures of connectivity in hedgegraphs and study their structural and algorithmic aspects. Instead of the cut function, we investigate a polymatroid associated with hedgegraphs. The polymatroidal lens leads to new tractability results as well as insightful generalizations of classical results on graphs and hypergraphs. 
\end{abstract}


\pagenumbering{arabic}

\section{Introduction}\label{section:introduction}
Graphs and hypergraphs are fundamental objects in discrete and combinatorial optimization. They provide a powerful way to model many problems that arise in applications. A hypergraph $G=(V, E)$ is specified by a vertex set $V$ and a collection $E$ of hyperedges where each hyperedge $e\in E$ is a subset of vertices. If every hyperedge has size at most two, then the hypergraph is a graph.
Hedgegraphs generalize hypergraphs further --- a hedge is a collection of hyperedges in an underlying hypergraph. There are various scenarios where a collection of hyperedges (or edges in a graph) could be interdependent. For example, consider a supply-chain network where a certain collection of hyperedges could depend on the same resource and could all fail together if the resource fails; such dependencies are conveniently modeled by introducing a hedge/color that contains these hyperedges. 
Coudert, Datta, Perennes, Rivano, and Voge \cite{CDPRV07} introduced graphs with colored edges to model dependency between a collection of edges. The term hedgegraph for such colored graphs was coined by Ghaffari, Karger, and Panigrahi \cite{GKP17} who were motivated by applications in graph reliability where multiple edges can fail in a coordinated manner. 
While the richer modeling capability of hedgegraphs naturally lends itself to applications, it also poses algorithmic and structural challenges---see \cite{CDPRV07, ZCTZ11, BGGTU15, CPRV16, GKP17, CXY19, AVB20, FPZ23, JLMPS23, PST24, FGKLS25}. In fact, as we will discuss later, some of the fundamental tractable problems in graphs and hypergraphs become intractable in hedgegraphs.
We investigate structural and algorithmic aspects of connectivity in hedgegraphs by taking a polymatroidal viewpoint that leads to new tractability results as well as insightful connections to classical results on graphs \cite{Tut61, NW61, NW64} and more recent results on hypergraphs \cite{FKK03,FKK03-ori}.

We formally define hedgegraphs now. A hedgegraph $G=(V,E)$ consists of a finite vertex set $V$ and a finite set $E$ of \emph{hedges}. 
A \emph{hyperedge} over $V$ is a subset of vertices. Each \emph{hedge} $e \in E$ is a collection (equivalently a set) of distinct hyperedges; by distinct we mean that a hyperedge $h$ in a hedge $e$ does not belong to another hedge $e' \in E$ where $e'\neq e$. 
If every hedge has exactly one hyperedge, then the hedgegraph is simply a hypergraph.  Throughout this work, we will assume that each hedge is a collection of vertex-disjoint hyperedges: if we have a hedge $e$ with hyperedges $h_1, h_2\in e$ where $h_1\cap h_2\neq \emptyset$, then we remove $h_1$ and $h_2$ from $e$ and add the hyperedge $h_1\cup h_2$ to $e$; we will later see that such a replacement does not affect cut and partition capacities which are the quantities of interest to this work. This assumption about each hedge being a vertex-disjoint collection of hyperedges is crucial for a concise description of our structural results. Moreover, it abstracts the core aspects of the hedgegraph definition and allows us to easily relate hedgegraphs to hypergraphs\footnote{Hedgegraphs in \cite{GKP17} were defined via colored edges of a graph. Our definition is equivalent to their definition for the purposes of the problems considered in this work.}.
The representation size of a hedgegraph is $O(p)$, where $p:=\sum_{e\in E}\sum_{h\in e}|h|$. 
See Figure \ref{figure:hedgegraph} for an example.
\begin{figure}[h]
    \centering
    \includegraphics[width=1.0\linewidth]{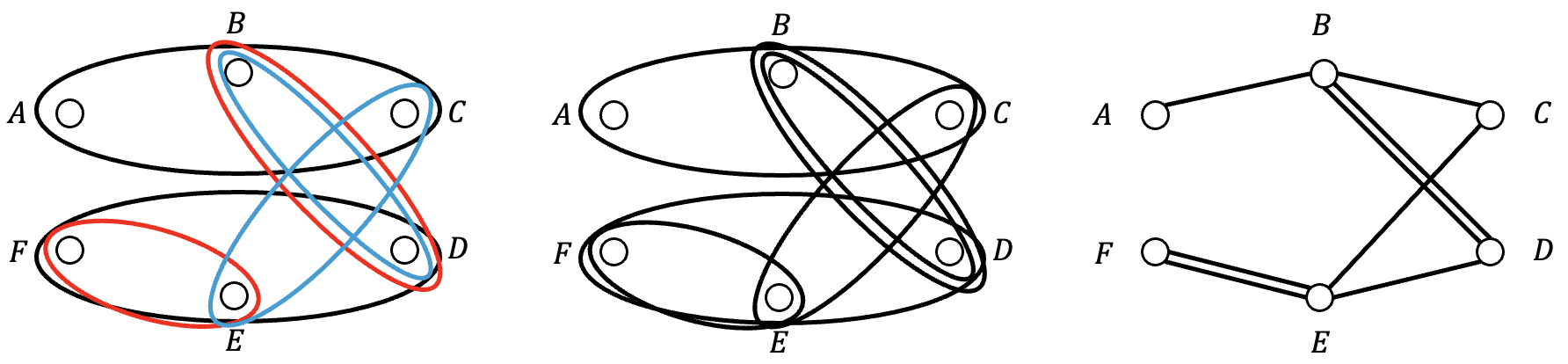}
    \caption{Leftmost figure shows a hedgegraph with $V = \{A,B,C,D,E,F\}$ and three hedges, where the black hedge $e_1:=\{\{A,B,C\}, \{D,E,F\}\}$, the red hedge $e_2:=\{\{B,D\},\{E,F\}\}$, and the blue hedge $e_3:=\{\{C,E\}, \{B,D\}\}$. It can also be viewed as a coloring of hyperedges of the hypergraph shown in the middle. Rightmost figure shows an example of a graph.
    }
    \label{figure:hedgegraph}
\end{figure}

We now define cuts and connectivity in hedgegraphs. Let $G=(V, E)$ be a hedgegraph. 
For a subset $S\subseteq V$ of vertices, we define $\delta_G(S)$ to be the set of hedges $e\in E$ that cross $S$; a hedge $e$ crosses $S$ if there is a hyperedge $h\in e$ that intersects both $S$ and $V\setminus S$. Formally,
\[
\delta_G(S):=\left\{e\in E: \exists\ h\in e \text{ with } h\cap S\neq \emptyset, h\cap (V\setminus S)\neq \emptyset\right\}. 
\]
We define the cut function $d_G: 2^V\rightarrow \Zg$ of the hedgegraph $G$ as $d_G(S):=|\delta_G(S)|$ for every $S\subseteq V$. The \emph{connectivity} of $G$ is defined as $\lambda_G:=\min\{d_G(S): \emptyset\neq S\subsetneq V\}$. The hedgegraph $G$ is \emph{connected} if $\lambda_G>0$. 
We drop the subscript $G$ from $\delta_G, d_G, \lambda_G$ when the hedgegraph $G$ is clear from context. It is possible to verify whether a given hedgegraph is connected in polynomial time. 

In contrast to graph and hypergraph cut functions, the cut function of a hedgegraph is not necessarily submodular --- this was pointed out in \cite{GKP17} via the example in Figure~\ref{figure:hedge-cut-function}. The lack of submodularity creates barriers for structure and algorithms. In particular, $\{s,t\}$-connectivity\footnote{$\{s,t\}$-connectivity in a hedgegraph is defined as $\min\{d(S): s\in S\subseteq V-t\}$ , i.e., it is the minimum number of hedges whose deletion disconnects $s$ and $t$.} and connectivity are intractable unders plausible complexity theoretic assumptions \cite{ZF16, ZCTZ11, FGK10, JLMPS23}. Despite the intractability, positive algorithmic results are known for computing connectivity---in particular, a randomized polynomial-time approximation scheme, a randomized quasi-polynomial time algorithm \cite{GKP17, FPZ23}, a randomized polynomial-time algorithm if every hedge has constant number of hyperedges \cite{CXY19}, and a fixed-parameter algorithm when parameterized by the solution size \cite{FGKLS25} are known. All these algorithms are randomized and it remains open to design deterministic counterparts. Inspired by these algorithmic results, we investigate structural aspects of connectivity in hedgegraphs and some of their algorithmic consequences. We believe that our structural results will provide valuable insights for future work on hedgegraphs. 

\begin{figure}[h]
    \centering
    \includegraphics[width=0.3\linewidth]{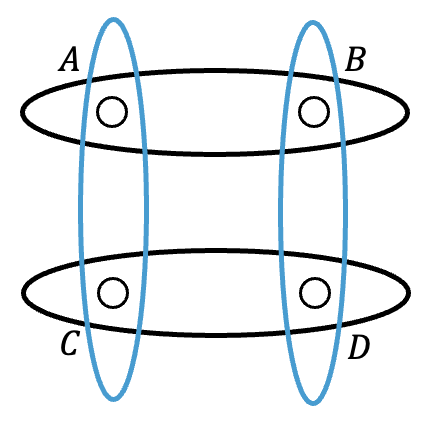}
    \caption{An example of a hedgegraph $G=(V, E)$ whose cut function is not submodular. Here, $V=\{A,B,C,D\}$ 
    and $E:=\{e_1, e_2\}$ where $e_1:=\{\{A,B\},\{C,D\}\}$ and $e_2=\{\{A,C\},\{B,D\}\}$. We note that $d(\{A,B\})=d(\{A,C\})=1$ and $d(\{A\})=d(\{A,B,C\})=2$ which implies that $d(\{A,B\})+d(\{A,C\})<d(\{A\})+g(\{A,B,C\})$.}
    \label{figure:hedge-cut-function}
\end{figure}

\paragraph{Partition-based Measures of Connectedness.} 
Given that connectivity in hedgegraphs is intractable \cite{JLMPS23}, are there other measures of connectedness that are relevant and interesting? We investigate two alternative measures of connectedness of a hedgegraph based on partitions. 
For ease of understanding, we define and discuss these measures for hypergraphs here and postpone the formal definition for hedgegraphs to Section \ref{sec:results}. Let $G=(V, E)$ be a hypergraph.
For a partition $\mcp$ of $V$, we define $\delta_G(\mcp)$ to be the set of hyperedges that intersect more than one part of the partition $\mcp$, i.e., 
\[
\delta_G(\mcp):=\{h\in E: \exists\ A, B\in \mcp\text{ with }h\cap A\neq \emptyset \text{ and } h\cap B\neq \emptyset\}.
\]
We denote \emph{partition connectivity} and \emph{weak partition connectivity} of $G$ by $\pc_G$ and $\wpc_G$ respectively, where 
\begin{align}
    \pc_G&:=\min_{\mathcal{P} \text{ is a partition of } V}\left \lfloor\frac{|\delta(\mathcal{P})|}{|\mathcal{P}|-1}\right \rfloor\text{ and}\label{eq:pc-hypergraphs}\\
    \wpc_G&:=\min_{\mathcal{P} \text{ is a partition of } V}\left\lfloor\frac{\sum_{e\in E}\left(\#\text{ of parts of }\mcp\text{ intersecting }e - 1\right)}{|\mathcal{P}|-1}\right\rfloor. \label{eq:wpc-hypergraphs}
\end{align}
These two measures coincide in graphs while they differ in hypergraphs. Both these measures are also related to connectivity: 
for a hypergraph $G$, we have that $\pc_G\le \wpc_G$ and $\lambda_G/2 \le \wpc_G\le \lambda_G$, i.e., weak partition connectivity is within a $2$ factor of connectivity \cite{BCKK24}. 
It is helpful to keep three examples in mind while understanding partition connectivity: consider the hypergraph $G_1=(V, E_1)$ on vertex set $V$ where $E_1:=\{V\}$, the hypergraph $G_2=(V, E_2)$ where $E_2$ consists of $|V|-1$ parallel hyperedges $e_1 = e_2 = \ldots= e_{|V|-1}=V$, and the cycle graph $C_3$ on $3$ vertices: we observe that $\pc_{G_1}=0$, $\wpc_{G_1}=1$, $\lambda_{G_1}=1$, $\pc_{G_2}=1$, $\wpc_{G_2}=|V|-1$, $\lambda_{G_2}=|V|-1$, $\pc_{C_3}=\wpc_{C_3}=1$, and $\lambda_{C_3}=2$. 
Tutte \cite{Tut61} and Nash-Williams \cite{NW61} showed that partition connectivity of a graph is the maximum number of edge-disjoint spanning trees in the graph. This result has far-reaching implications in graph decomposition and orientation \cite{Frank-book, West-book, Jae75, Jae79}. Frank, Kir\'{a}ly, and Kriesell \cite{FKK03} generalized Tutte and Nash-Williams' result on partition connectivity of graphs to hypergraphs by showing a hypergraph decomposition result. Frank, Kir\'{a}ly, and Kir\'{a}ly 
\cite{FKK03-ori} generalized Tutte and Nash-Williams' result on weak partition connectivity of graphs to hypergraphs by showing an orientation result. We observe that the connection between weak partition connectivity of hypergraphs and hypergraph orientations has recently found applications in coding theory \cite{GLSTW24, AGL24}. In this work, we generalize and explore these two partition based connectedness measures in the context of hedgegraphs.

\paragraph{Hedgegraph Polymatroid.}
Instead of the cut function, which lacks submodularity, we investigate an alternative function associated with a hedgegraph. We motivate this first. Several structural properties of graphs are derived via 
a matroidal view, in particular through the graphic matroid induced by the edges of the graph. In particular, Tutte and Nash-Williams' result on partition connectivity can be viewed as a special case of Edmonds' matroid base packing theorem applied to the graphic matroid. We recall that for a graph $G=(V, E)$, the matroid rank function $r: 2^E\rightarrow \Z_{\ge 0}$ is defined by $r(A):=|V|-\Comp(V,A)$ for every $A\subseteq E$, where $\Comp(V,A)$ is the number of connected components in the subgraph $(V, A)$. 
We may define a similar function for hypergraphs---for a hypergraph $H=(V, E)$, consider the function $f_H: 2^E\rightarrow \Z_{\ge 0}$ defined as $f_H(A):=|V|-\Comp(V, A)$ for every $A\subseteq E$, where $\Comp(V,A)$ is the number of connected components in the sub-hypergraph $(V, A)$. 
The function $f_H$ is not the rank function of a matroid. Nevertheless, the function $f_H$ is monotone and submodular, and hence, is a polymatroid. 
One can associate a matroid with every polymatroid over the same ground set (via Edmonds' construction of the rank function of a matroid from a polymatroid \cite{Edm70}) and investigate the structural properties of the resulting matroid to understand structural properties of the polymatroid. This viewpoint is implicit in \cite{FKK03} who study the hypergraphic matroid of Lorea \cite{Lorea78} which is equivalent to the matroid derived from the function $f_H$ mentioned above via Edmonds' construction of the rank function. In this work, we 
make this viewpoint explicit and use the polymatroid perspective 
to investigate partition connectivity, weak partition connectivity, and other structural aspects of hedgegraphs. 

We now define the relevant polymatroid that underlies our investigation of hedgegraphs. 
Let $G=(V, E)$ be a hedgegraph. 
The number of connected components in a hedgegraph $G$ is the number of connected components in the hypergraph (possibly multi-hypergraph) $(V, H_E)$, where $H_E:=\{h: h\in e\text{ for some }e\in E\}$ of hyperedges.
For a subset $A\subseteq E$  of hedges, we define $(V, A)$ to be the sub-hedgegraph over the vertex set $V$ that contains all hedges in $A$ 
and $\Comp(V,A)$ to denote the number of connected components in the sub-hedgegraph $(V, A)$. We define the function $f_G: 2^E\rightarrow \Zg$ as 
    \[
    f_G(A) := |V| - \Comp(V,A)\ \forall\ A\subseteq E. 
    \]
We observe that the function $f_G$ is monotone, i.e., $f_G(A)\le f_G(B)$ for every $A\subseteq B\subseteq E$ and submodular, i.e., $f_G(A) + f_G(B) \ge f_G(A\cap B) + f_G(A\cup B)$---see Appendix~\ref{appendix:hedgegraph-polymatroid} for a formal proof. Hence, we denote $f_G$ as the \emph{hedgegraph polymatroid} associated with the hedgegraph $G$. 
We emphasize that the cut function is a set function defined over subsets of vertices while the hedgegraph polymatroid is a set function defined over subsets of hedges. If the hedgegraph is a graph, then the hedgegraph polymatroid is simply the rank function of the graphic matroid associated with the graph. The hedgegraph polymatroid is the protagonist of our work that helps generalize several structural and algorithmic aspects of graphs and hypergraphs to hedgegraphs. In order to understand the two partition based measures of connectedness of hedgegraphs, we rely on two notions of strengths of a polymatroid. These two notions of strengths of a polymatroid have found interesting applications recently \cite{CCZ25, quanrud2024quotient}.

\subsection{Results}\label{sec:results}
We define the two partition based measures of connectedness for hedgegraphs and describe our results below. Subsequently, we describe our results for hedgegraph cut sparsification. 

\paragraph{Partition Connectivity.}
We define the \emph{partition connectivity} of a hedgegraph $G=(V, E)$ as
$$\pc_G:=\min_{\mathcal{P} \text{ is a partition of } V}\left \lfloor\frac{|\delta(\mathcal{P})|}{|\mathcal{P}|-1}\right \rfloor,$$
where $\delta_G(\mathcal{P}):=\{e\in E: \exists\ h\in e\text{ with } h \text{ intersecting at least two parts of }\mathcal{P}\}$. A hedgegraph $G$ is \emph{$k$-partition-connected} if $\pc_G\ge k$. 
We show that partition connectivity of a given hedgegraph can be computed in polynomial time.
\begin{restatable}{theorem}{theorempartitionconnectivity}\label{theorem:partition-connectivity}
    There exists a deterministic polynomial-time algorithm to compute the partition connectivity of a given hedgegraph. 
\end{restatable}


Tutte \cite{Tut61} and Nash-Williams \cite{NW61} showed that partition connectivity in graphs measures decomposition into spanning trees: a graph is $k$-partition connected if and only if it has $k$ edge-disjoint spanning trees. 
There is an alternative way to phrase this result that enables a generalization to hypergraphs. For this, we note that a graph is $1$-partition connected if and only if it contains a spanning tree. Hence, Tutte and Nash-Williams' result is equivalent to saying that a graph is $k$-partition connected if and only if it contains $k$ edge-disjoint $1$-partition connected subgraphs. Frank,  Kir\'{a}ly, and Kriesell \cite{FKK03} showed that this statement also holds for hypergraphs: a hypergraph is $k$-partition connected if and only if it contains $k$ hyperedge-disjoint $1$-partition connected sub-hypergraphs. We show that this statement also holds for hedgegraphs under our generalized definition of partition connectivity in hedgegraphs.


\begin{restatable}{theorem}{theoremdecomposition}\label{theorem:hedgegraph-decomposition}
     A hedgegraph is $k$-partition connected if and only if it contains $k$ hedge-disjoint $1$-partition connected sub-hedgegraphs. 
    
\end{restatable}

As stated earlier, a graph is $1$-partition connected if and only if it contains a spanning tree. \cite{FKK03} characterized $1$-partition connected hypergraphs via a trimming result. A \emph{trimming of a hyperedge} $h$ is obtained by picking a pair of vertices $u, v\in h$ and replacing the hyperedge $h$ by the edge $\{u, v\}$. A \emph{trimming of a hypergraph} $G=(V, H)$ is obtained by trimming each hyperedge $h\in H$. \cite{FKK03} showed that a hypergraph $G$ is $1$-partition connected if and only if it contains a sub-hypergraph that can be trimmed to a spanning tree. We generalize this result to hedgegraphs. 
A \emph{trimming of a hedge} is obtained by picking a hyperedge $h\in e$, a pair of vertices $u, v\in h$, and replacing all hyperedges in $e$ by the single edge $\{u, v\}$. A \emph{trimming of a hedgegraph} $G=(V, E)$ is obtained by trimming each hedge $e\in E$. We emphasize that trimming does not allow complete deletion of hedges---in particular, every trimming of a hedgegraph with $m$ hedges results in a graph with exactly $m$ edges (possibly multigraph). 

\begin{restatable}{theorem}{theoremtrimming}\label{theorem:trimming}
    A hedgegraph is $1$-partition connected if and only if it contains a sub-hedgegraph that can be trimmed to a spanning tree. 
\end{restatable}
We will see a proof of Theorem \ref{theorem:partition-connectivity} (i.e., an efficient algorithm to compute partition connectivity of a given hedgegraph) by reduction to submodular minimization. 
We mention that our techniques underlying the proof of Theorems \ref{theorem:hedgegraph-decomposition} and \ref{theorem:trimming} give an alternative algorithm to compute partition connectivity via matroid intersection. 

Our techniques also help generalize the notion of arboricity from graphs and hypergraphs to hedgegraphs. Let $k\in \Z_+$. 
Nash-Williams \cite{NW64} showed that the edge set $E$ of a graph can be partitioned into $k$ forests if and only if $|E[S]|\le k(|S|-1)$ for every $S\subseteq V$. Frank, Kir\'{a}ly, and Kriesell \cite{FKK03} generalized this statement to hypergraphs through the notion of acyclic-trimmability: a set of hyperedges has an \emph{acyclic-trimming} if it has a trimming that is acyclic. They showed that the hyperedge set $E$ of a hypergraph can be partitioned into $k$ acyclic-trimmable subsets if and only if $|E[S]|\le k(|S|-1)$ for every $S\subseteq V$, where $E[S]$ is the subset of hyperedges that are fully contained in $S$. 
We generalize this statement to hedgegraphs next. For a partition $\mathcal{P}$ of the vertex set $V$ of a hedgegraph $G=(V, E)$, we define $E[\mathcal{P}]:=E\setminus \delta(\mathcal{P})$. A set of hedges has an \emph{acyclic-trimming} if it has a trimming that is acyclic. See Figure~\ref{figure:trimming} for an example.

\begin{figure}[h]
    \centering
    \includegraphics[width=0.7\linewidth]{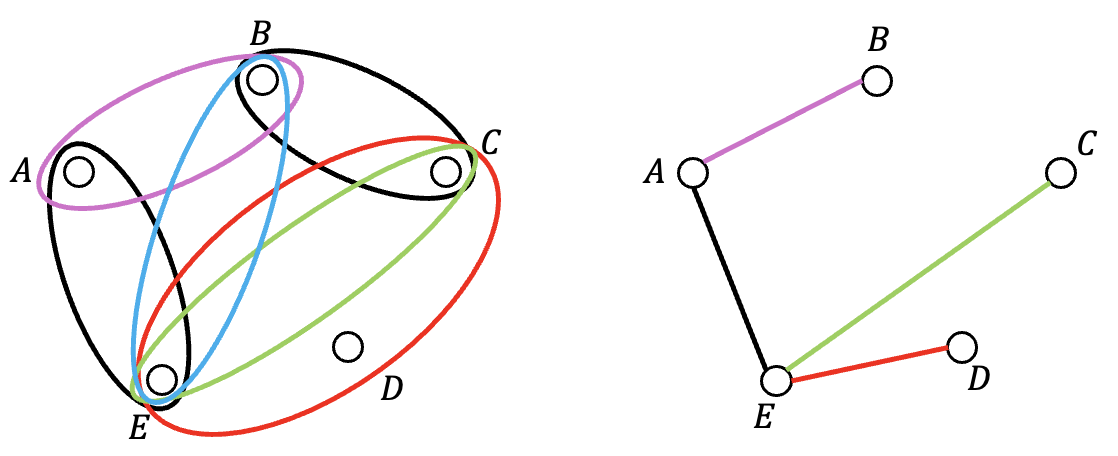}
    \caption{An example of an acyclic-trimming. The left side shows a hedgegraph $G=(V,E)$ with five hedges, where $e_1:=\{\{A,E\},\{B,C\}\}$, $e_2:=\{\{A,B\}\}$, $e_3:=\{\{C,E\}\}$, $e_4:=\{\{C,D,E\}\}$, and $e_5:=\{\{B,E\}\}$. The right side shows an acyclic-trimming of $\{e_1, e_2, e_3, e_4\}$, where $e_1$ is trimmed into $\{A,E\}$, $e_2$ is trimmed into $\{A,B\}$, $e_3$ is trimmed into $\{C,E\}$, and $e_4$ is trimmed into $\{D,E\}$. We may verify that $\{e_1, e_2, e_3, e_5\}$ does not have an acyclic-trimming.}
    \label{figure:trimming}
\end{figure}

\begin{restatable}{theorem}{theoremcovering}\label{theorem:base-covering}
Let $G=(V, E)$ be a hedgegraph and $k\in \Z_{+}$. Then, $E$ can be partitioned into $k$ acyclic-trimmable subsets if and only if $|E[\mathcal{P}]|\le k(|V|-|\mathcal{P}|)$ for every partition $\mathcal{P}$ of $V$. 
\end{restatable}
If $G$ is a hypergraph/graph, then $|E[\mathcal{P}]|\le k(|V|-|\mathcal{P}|)$ for every partition $\mathcal{P}$ of $V$ if and only if $E[S]\le k(|S|-1)$ for every subset $S\subseteq V$. Thus, the partition based condition is equivalent to subset based condition in graphs and hypergraphs. 
However, this equivalence does \emph{not} hold for hedgegraphs. We emphasize that our techniques underlying the proof of Theorem \ref{theorem:base-covering} also enables a polynomial-time algorithm to find the least $k$ such that the hedge set of a given hedgegraph can be partitioned into $k$ acyclic-trimmable subsets. 

\paragraph{Implications for Orientations.}
Next we state another consequence of Theorems \ref{theorem:hedgegraph-decomposition} and \ref{theorem:trimming}. 
An \emph{orientation of an undirected graph} is a digraph obtained by picking a head vertex for each edge. A digraph $\overrightarrow{G}$ with a specified root $r\in V(\overrightarrow{G})$ is \emph{rooted $k$-out-arc-connected} if $d^{out}_{\overrightarrow{G}}(U)\ge k$ for every $r\in U\subsetneq V$. 
We recall that Tutte and Nash-Williams' min-max relation for partition connectivity can also be phrased as an orientation result: a graph is $k$-partition connected if and only if it has a rooted $k$-out-arc-connected orientation for some choice of the root vertex $r$. 
Theorems \ref{theorem:hedgegraph-decomposition} and \ref{theorem:trimming} lead to a similar orientation result for hedgegraphs that we now state. A \emph{directed hypergraph} is specified by the tuple $\overrightarrow{G}=(V, H, head: H\rightarrow V)$, where $V$ is the vertex set, $H$ is a collection of hyperedges, and $head(h)\in h$ for every hyperedge $h\in H$; $head(h)$ is said to be the head vertex of the hyperedge $h$. 
An \emph{orientation of a hedgegraph} $G=(V, E)$ is a directed hypergraph obtained by picking \emph{exactly} one hyperedge $h\in e$ for each $e\in E$ and orienting that hyperedge (i.e., picking a head vertex for that hyperedge); we note that the other hyperedges in the hedge are discarded.
A directed hypergraph $\overrightarrow{G}$ with a specified root $r\in V(\overrightarrow{G})$ is \emph{rooted $k$-out-hyperarc-connected} if $d^{out}_{\overrightarrow{G}}(U)\ge k$ for every $r\in U\subsetneq V$, where $d^{out}_{\overrightarrow{G}}(U):=|\{h\in E(\overrightarrow{G}): head(h)\in V\setminus U, h\cap U\neq \emptyset\}|$ for every $U\subseteq V$. See Figure~\ref{figure:orientation} for an example. By Menger's theorem, a rooted $k$-out-hyperarc-connected hypergraph has $k$ hyperedge disjoint paths from root to each vertex. 

\begin{figure}[h]
    \centering
    \includegraphics[width=0.5\linewidth]{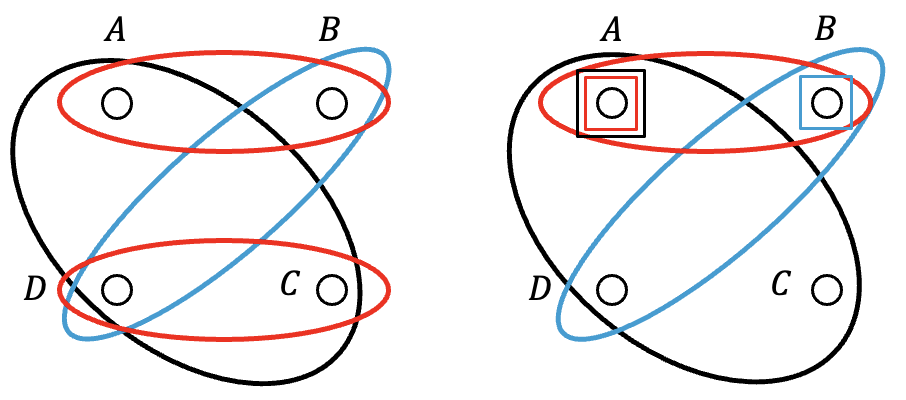}
    \caption{An example of hedgegraph orientation. The left side shows a hedgegraph $G=(V,E)$ with three hedges, where $e_1:=\{\{A,C,D\}\}$, $e_2:=\{\{A,B\},\{C,D\}\}$, and $e_3:=\{\{B,D\}\}$. The right side shows an orientation of $G$, where hyperedge $\{A,C,D\}$ is picked in $e_1$ with $A$ being its head, hyperedge $\{A,B\}$ is picked in $e_2$ with $A$ being its head, and hyperedge $\{B,D\}$ is picked in $e_3$ with $B$ being its head. If we let $r:=A$ being the root, then it is a rooted $1$-out-hyperarc-connected hypergraph. We note that this is not rooted $2$-out-hyperarc-connected since $d^{out}_{\overrightarrow{G}}(\{A,B,D\})=d^{out}_{\overrightarrow{G}}(\{A,C,D\})=1$.}
    \label{figure:orientation}
\end{figure}
\begin{restatable}{corollary}{corollaryorientation}\label{corollary:orientation}
    A hedgegraph is $k$-partition connected if and only if it has a rooted $k$-out-hyperarc connected orientation for some choice of the root vertex. 
\end{restatable}

We emphasize that for hedgegraphs, existence of a rooted $k$-out-hyperarc connected orientation for some choice of the root vertex is equivalent to the existence of a rooted $k$-out-hyperarc connected orientation for every choice of the root vertex. We also note that hypergraphs and hedgegraphs differ substantially from graphs for rooted connected orientation problems---in particular, the existence of a rooted $k$-in-hyperarc connected orientation of a hypergraph\footnote{A directed hypergraph $\overrightarrow{G}$ with root $r\in V(\overrightarrow{G})$ is \emph{rooted $k$-in-hyperarc-connected} if $d^{in}_{\overrightarrow{G}}(U)\ge k$ for every $r\in U\subsetneq V$, where $d^{in}_{\overrightarrow{G}}(U):=|\{h\in E(\overrightarrow{G}): head(h)\in U, h\setminus U\neq \emptyset\}|$. } does not imply the existence of a rooted $k$-out-hyperarc connected orientation although this fact holds trivially in graphs (consider reversing the orientation of all arcs).

\paragraph{Weak Partition Connectivity.}
We define the weak partition connectivity of a hedgegraph $G=(V, E)$ as
\[
\wpc_G:=\min_{\mathcal{P} \text{ is a partition of } V}\left\lfloor\frac{\sum_{e\in E}\left(|\mathcal{P}|-\Comp(\mathcal{P}(e))\right)}{|\mathcal{P}|-1}\right\rfloor,
\] 
where $|\mathcal{P}|$ is the number of parts of the partition $\mathcal{P}$ and $\mathcal{P}(e)$ is the hedgegraph obtained from the hedgegraph $(V,\{e\})$ by contracting every part in $\mathcal{P}$ into a single vertex. See Figure~\ref{figure:weak-partition-connectivity} for an example.
\begin{figure}[h]
    \centering
    \includegraphics[width=0.6\linewidth]{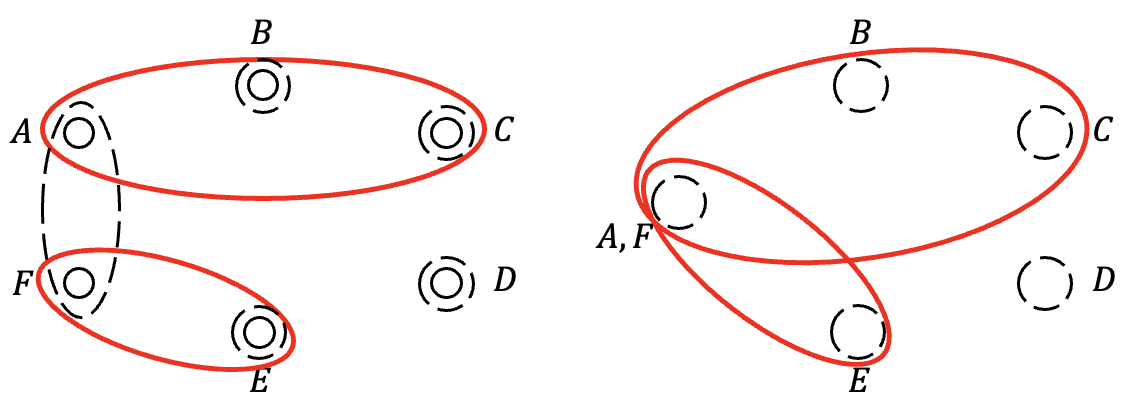}
    \caption{An example of vertex partition $\mathcal{P}$ and hedgegraph $\mathcal{P}(e)$. The left side shows a partition $\mathcal{P}:=\{\{B\}, \{C\},\{D\},\{E\},\{A,F\}\}$ of vertex set $V$ and a hedge $e:=\{\{A,B,C\},\{E,F\}\}$. The right side shows the hedgegraph obtained from the hedgegraph $(V,\{e\})$ by contracting every part in $\mathcal{P}$ into a single vertex. We note that $|\mathcal{P}|=5$ and $\Comp(\mathcal{P}(e))=2$.}
    \label{figure:weak-partition-connectivity}
\end{figure}
As mentioned before, partition connectivity and weak partition connectivity are equal in graphs but differ in hypergraphs. 
Our definition of weak partition connectivity in hypergraphs coincides with the definition in (\ref{eq:wpc-hypergraphs}) although they might seem different on first glance. 

It is well-known that partition connectivity is within a $2$-factor of connectivity in graphs. Moreover, weak partition connectivity is within a $2$-factor of connectivity in hypergraphs \cite{BCKK24}. We show that this fact holds in hedgegraphs for the above-mentioned definition of weak partition connectivity. 
\begin{restatable}{lemma}{lemmawpcandconnectivity}\label{lemma:wpc-and-connectivity} 
Let $G=(V, E)$ be a hedgegraph with connectivity $\lambda>0$ and $\wpc_G$ be its weak partition connectivity. Then, $\left \lfloor\frac{\lambda}{2}\right \rfloor\le \wpc_G \leq \lambda$.
\end{restatable}

We now state two consequences of Lemma \ref{lemma:wpc-and-connectivity}. The first consequence addresses deterministic connectivity computation in hedgegraphs. 
Weak partition connectivity in graphs and hypergraphs can be computed in polynomial time \cite{FKK03-ori}. 
If weak partition connectivity in hedgegraphs can be computed in deterministic polynomial time, then it immediately implies a deterministic $2$-approximation for connectivity via Lemma \ref{lemma:wpc-and-connectivity}. 
At this point, we do not know if weak partition connectivity in hedgegraphs can be computed in polynomial time. 
However, we show that there exists a deterministic polynomial-time algorithm to compute a $\log$-approximation to connectivity using Lemma \ref{lemma:wpc-and-connectivity} and connections to the hedgegraph polymatroid. To the best of the authors' knowledge, this is the first deterministic algorithm with non-trivial approximation guarantee for connectivity in hedgegraphs. 
\begin{restatable}{theorem}{theoremlogapproxconnectivity}\label{theorem:log-approx-connectivity}
    There exists a deterministic polynomial-time algorithm to compute an $O(\log{n})$-approximation for connectivity in hedgegraphs. 
\end{restatable}

The second consequence of Lemma \ref{lemma:wpc-and-connectivity} addresses sampling in hypergraphs and hedgegraphs. 
A fundamental result in graphs is that sampling every edge of a $n$-vertex connected graph independently with probability  $\Omega(\log{n}/\lambda)$ gives a connected graph with high probability. This result is the basis for understanding graph reliability under random edge failures \cite{karger1994random, Kar-FPRAS-unreliability}.  
This sampling result for graphs was first shown by Karger \cite{karger1994random} using the fact that the number of $\alpha$-approximate min-cuts  in a $n$-vertex graph is $n^{O(\alpha)}$. It is well-known that the bound on the number of approximate min-cuts does not hold for hypergraphs---e.g., see \cite{CXY19}. Hence, it was unclear whether the same sampling result would hold for hypergraphs and hedgegraphs. 
Ghaffari, Karger, and Panigrahi \cite{GKP17} showed a slightly weaker bound for hedgegraphs (and hence, hypergraphs): sampling every hedge with probability $\Omega(\log^2{n}/\lambda)$ gives a connected hedgegraph with high probability. We improve on their result on hedgegraphs to match the guarantee on graphs. 

\begin{restatable}{theorem}{theoremrandomsampling}\label{theorem:random-sampling}
    Let $G=(V,E)$ be a $n$-vertex hedgegraph with connectivity $\lambda>0$. Let $A\subseteq E$ be a subset obtained by sampling each hedge $e\in E$ with probability at least $\frac{20\log n}{\lambda}$. Then, the subhedgegraph $(V, A)$ is connected with probability at least $1-2/n$.
\end{restatable}

We have stated our results above for unweighted hedgegraphs in the interests of notational brevity. All of the above results results can be extended to hedge-weighted hedgegraphs in a natural manner albeit with additional notational clutter. 

\paragraph{Partition Sparsifiers.} Cut and spectral sparsifiers in graphs are extensively studied and highly influential \cite{karger1994using, karger1995random, benczur1996approximating, karger1994random, karger2000minimum, karger2015fast, fung2011general, spielman2004nearly, spielman2011spectral, spielman2008graph, BSS12, LS18}. These have been extend to hypergraphs recently \cite{kogan2015sketching, chekuri2018minimum, chen2020near, soma2019spectral, bansal2019new, kapralov2022spectral, kapralov2021towards, jambulapati2023chaining, khanna2024near, khanna2025near, khanna2025sketch}.
In the setting of hedgegraphs we are concerned with partition sparsifiers which, as we argued, have useful structural properties. To describe the notion of partition sparsifiers, we need to consider weighted hedgegraphs. Let $G=(V, E)$ be a hedgegraph with hedge-weights $w: E\rightarrow \Rg$. 
For a partition $\mathcal{P}$ of $V$, we define 
the partition capacity as $d_{(G,w)}(\mathcal{P}):=\sum_{e\in \delta(\mathcal{P})} w(e)$. 
A weighting $w': E \rightarrow \Rg$ is a partition sparsifier for $(G, w)$ if all partition capacities are preserved within a $\pm \epsilon$ factor, i.e., $d_{(G, w')}(\mathcal{P})\in [(1-\epsilon)d_{(G, w)}(\mathcal{P}), (1+\epsilon)d_{(G, w)}(\mathcal{P})]$ for every partition $\mathcal{P}$ of $V$ and support$(w')$ is small. 
For graphs, cut sparsifiers and partition sparsifiers are equivalent (since preserving cut capacities suffices to preserve partition capacities). 
It is well-known that every weighted graph $(G, w)$ has a cut sparsifier $w'$ where support$(w')=O(n/\varepsilon^2)$ and it can be constructed in near-linear time, where $n$ is the number of vertices in $G$ \cite{BSS12, LS18}. 
For hypergraphs, cut sparsifiers are not equivalent to partition sparsifiers (since approximating cut capacities is not equivalent to approximating partition capacities).  
Quanrud \cite{quanrud2024quotient} showed that every weighted hypergraph $(G,w)$ has a partition sparsifier $w'$ where support$(w')=O(n\log{n}/\varepsilon^2)$ and it can be constructed in near-linear time. In fact, Quanrud proved a more general result regarding polymatroid quotient sparsification. We show a counting result for hedgegraphs that allows us to apply his quotient sparsification result for the hedgegraph polymatroid to conclude the following result on partition sparsifiers for hedgegraphs. 


\begin{restatable}{theorem}{theoremsparsification}\label{theorem:sparsification}
    Let $G=(V, E)$ be a hedgegraph with hedge-weights $w: E \rightarrow \Rg$ on $n$ vertices of size $p:=\sum_{e\in E}\sum_{h\in e}|h|$. 
    Then, there exist hedge-weights $w': E\rightarrow \Rg$ such that (1) support$(w')=O(n\log{n}/\varepsilon^2)$ and (2) $d_{(G, w')}(\mathcal{P})\in [(1-\varepsilon)d_{(G, w)}(\mathcal{P}), (1+\varepsilon)d_{(G, w)}(\mathcal{P})]$ for every partition $\mathcal{P}$ of $V$.
    Moreover, such a weighting $w'$ can be computed with high probability in polynomial time.
\end{restatable}

Quanrud \cite{quanrud2024quotient} further showed that a partition sparsifier for hypergraphs can be computed in near-linear runtime via a push-relabel algorithm. We believe that the same ideas would also extend to hedgegraphs.
Partition sparsification has natural algorithmic applications in reducing the size of the hedgegraph for various problems. We do not explore concrete applications of partition sparsifiers in this work.
\subsection{Related Work}\label{sec:related-work}
As mentioned before, $\{s,t\}$-connectivity is intractable: it is NP-hard to compute even if each hedge consists of two edges \cite{ZF16}, NP-hard to approximate within a factor of $2^{\log^{1-1/\log{\log^c{n}}}{n}}$ for every constant $c<1/2$ where $n$ is the number of vertices in the input hedgegraph \cite{ZCTZ11}, and W[2]-hard when parameterized by $\lambda$ \cite{FGK10}. 
Moreover, connectivity is ETH-hard to compute in time $(np)^{o(\log{n}/(\log{\log{n}})^2)}$, where $p$ is the representation size of the hedgegraph \cite{JLMPS23}. 

Despite the above-mentioned hardness results, certain positive results are known about the problem of computing connectivity. There exists a randomized polynomial-time approximation scheme as well as a randomized quasi polynomial-time algorithm, i.e. a randomized algorithm that runs in time $n^{O(\log{\lambda})}$ \cite{GKP17, FPZ23}. For hedgegraphs in which the number of hedges with at least two hyperedges is at most $m_0$, there exists a $2^{m_0} \cdot n^{O(1)}$-time deterministic algorithm \cite{CPRV16}. For hedgegraphs in which each hedge has at most a constant number of hyperedges, there exists a randomized polynomial time algorithm \cite{CXY19}. From the FPT perspective, there exists a randomized algorithm to compute connectivity in time $2^{\lambda}(n+m)^{O(1)}$, where $m$ is the number of hedges \cite{FGKLS25}. Some of these algorithmic results also extend to hedgegraph $k$-cut for fixed $k$: the goal in hedgegraph $k$-cut is to find a partition of the vertex set of a hedgegraph into $k$ non-empty parts $V_1, V_2, \ldots, V_k$ in order to minimize the number of hedges with at least one hyperedge intersecting distinct parts. 

\subsection{Preliminaries}\label{section:preliminaries}

\paragraph{Hedgegraphs.} Let $G=(V,E)$ be a hedgegraph. Throughout this work, we work under a simplifying assumption:  for every hedge $e\in E$, all hyperedges contained in $e$ are vertex-disjoint. 
If two hyperedges $h_1, h_2 \in e$ contained in a hedge $e$ share a vertex, i.e., $h_1 \cap h_2 \neq \emptyset$, then we can replace them with the hyperedge $h:=h_1 \cup h_2$ where the new hyperedge $h$ is contained in $e$. 
This operation does not change connectivity, partition connectivity, and weak partition connectivity. 

\paragraph{Matroids.}
A \emph{matroid} is defined by a tuple $\mathcal{M}=(\mathcal{N}, \mathcal{I})$, where $\mathcal{N}$ is a finite set and $\mathcal{I}\subseteq 2^{\mathcal{N}}$ is a collection satisfying the following two axioms: (1) if $B\in \mathcal{I}$ and $A\subseteq B$, then $A\in \mathcal{I}$ and (2) if $A, B\in \mathcal{I}$ and $|B|>|A|$, then there exists $e\in B\setminus A$ such that $A\cup\{e\}\in \mathcal{I}$. The rank function $r:2^\mathcal{N}\rightarrow \Z_{\ge 0}$ of a matroid $\mathcal{M}=(\mathcal{N}, \mathcal{I})$ is defined by $r(U):=\max\{|A|: A\subseteq U, A\in \mathcal{I}\}$ for every $U\subseteq \mathcal{N}$. We refer the reader to \cite{Frank-book} for further background on matroids. 

\paragraph{Polymatroids.} A \emph{polymatroid} $f:2^{\cN}\rightarrow \Zg$ on a ground set $\cN$ is an integer-valued monotone submodular function with $f(\emptyset)=0$. A function $f: 2^{\cN}\rightarrow \Zg$ is \emph{monotone} if $f(A)\leq f(B)$ for every $A\subseteq B$ and \emph{submodular} if $f(A)+f(B)\geq f(A\cup B)+f(A\cap B)$ for every $A,B\subseteq \cN$. A subset $S\subseteq \mathcal{N}$ is a \emph{base} of the polymatroid if $f(S)=f(\mathcal{N})$. 
For a weight function $w: \mathcal{N}\rightarrow \mathbb{R}_{\ge 0}$ and a subset $A\subseteq \mathcal{N}$, we define $w(A):=\sum_{e\in A}w(e)$. 
We will need two different notions of strength in a polymatroid. 

\begin{definition}
Let $f:2^{\mathcal{N}}\rightarrow \mathbb{Z}_{\geq 0}$ be a polymatroid. 
\begin{enumerate}
\item The \emph{functional-strength} of $f$ is (with the convention that $0/0=+\infty$): 
$$k^*(f):=\min_{A\subseteq \mathcal{N}} \left\lfloor \frac{\sum_{e\in \mathcal{N}}\left(f(A+e)-f(A)\right)}{f(\mathcal{N})-f(A)}\right\rfloor.$$
\item Let $w:\cN\rightarrow \mathbb{R}_{\geq 0}$ be weights on the elements. 
The \emph{$w$-strength of $f$} is (with the convention that $0/0=+\infty$): 
$$\kappa_{w}(f):=\min_{A\subseteq \mathcal{N}}\frac{w(\mathcal{N})-w(A)}{f(\mathcal{N})-f(A)}.$$
\end{enumerate}
\end{definition}

For matroid rank functions $f$, 
$k^*(f) = \lfloor \kappa_{\overrightarrow{\mathbf{1}}}(f)\rfloor$, where $\overrightarrow{\mathbf{1}}$ is the unit weight function, and moreover, both these quantities are equal to the maximum number of disjoint bases in the matroid \cite{Edm65-partition}. For polymatroids $f$, Călinescu, Chekuri, and Vondrák \cite{cualinescu2009disjoint} showed that $k^*(f)$ is within a $O(\log{f(\mathcal{N})})$-factor of the maximum number of disjoint bases. They also showed that $k^*(f)$ is NP-hard to approximate within a $o(\log{f(\mathcal{N})})$-factor while also giving a deterministic algorithm to approximate $k^*(f)$ within a $O(\log{f(\mathcal{N})})$-factor.

\begin{lemma}\cite{cualinescu2009disjoint}\label{lemma:functional-strength-deterministic}
    Given a polymatroid $f: 2^{\mathcal{N}}\rightarrow \mathbb{R}$ by its evaluation oracle, there is a deterministic polynomial time algorithm to compute an $O(\log f(\cN))$-approximation for $k^*(f)$. 
    
\end{lemma}

The following lemma shows that $\kappa_w(f)$ can be computed in polynomial time. 

\begin{lemma}\label{lemma:strength-in-polynomial-time}
    Given a polymatroid $f: 2^\mathcal{N}\rightarrow \mathbb{R}$ by its evaluation oracle and integer weights $w: \mathcal{N}\rightarrow \Zg$, there exists a strongly polynomial-time algorithm to compute $\kappa_w(f)$. 
\end{lemma}
\begin{proof}
    Consider the function $g:2^V\rightarrow \Z_{\ge 0}$ defined as $g(A):=f(\mathcal{N})-f(\mathcal{N}\setminus A)$ for every $A\subseteq \mathcal{N}$. Then, $\frac{w(\cN)-w(A)}{f(\cN)-f(A)} = \frac{w(\cN\setminus A)}{g(\cN\setminus A)}$ for every $A\subseteq \cN$. Thus, $\kappa_w(f) = \min_{A\subseteq \cN} \frac{w(\cN)-w(A)}{f(\cN)-f(A)} = \min_{A\subseteq \cN} \frac{w(\cN\setminus A)}{g(\cN\setminus A)}=\min_{B\subseteq \cN} \frac{w(B)}{g(B)}$. We observe that the function $g$ is monotone, submodular, and has $g(\emptyset)=0$. Consequently, the problem $\min_{B\subseteq \cN} \frac{w(B)}{g(B)}$ can be solved in strongly polynomial time via parametric submodular minimization \cite{Nag07}. 

\end{proof}

\section{Partition Connectivity of Hedgegraphs}\label{section:partition-connectivity}
In this section, we prove that partition connectivity of hedgegraphs can be computed in polynomial time.
The following lemma relates partition connectivity of a hedgegraph to the weighted strength of the hedgegraph polymatroid.

\begin{lemma}\label{lemma:partition-connectivity-equals-strength}
    Let $G=(V,E)$ be a connected hedgegraph and $f:2^E\rightarrow \Zg$ be the associated hedgegraph polymatroid.  
    Then, the partition connectivity of $G$ is 
    $ \pc_G=\lfloor\kappa_{\overrightarrow{\mathbf{1}}}(f)\rfloor$, where $\overrightarrow{\mathbf{1}}$ is the unit weight function.
\end{lemma} 
\begin{proof}

    We first show that $\pc_G\geq \lfloor\kappa_{\overrightarrow{\mathbf{1}}}(f)\rfloor$. Let $\mathcal{P}$ be a partition of $V$ such that $\pc_G=\lfloor\frac{|\delta(\mathcal{P})|}{|\mathcal{P}|-1}\rfloor$. In the subhedgegraph $(V,E\setminus \delta(\mathcal{P}))$, the number of connected components is at least $|\mathcal{P}|$, because for every hedge $e\in E\setminus \delta(\mathcal{P})$, there is no hyperedge in $e$ intersecting different parts of the partition $\mathcal{P}$. This implies that
    $$\begin{aligned}
        f(E)-f(E\setminus \delta(\mathcal{P})) &= (|V|-1) - (|V|-\Comp(V, E\setminus \delta(\mathcal{P}))) \ \ \text{(since $G$ is connected)}\\
        &= \Comp(V, E\setminus \delta(\mathcal{P})) - 1 \\
        &\geq |\mathcal{P}|-1.
    \end{aligned}$$
    Hence, 
    $$\begin{aligned}
        \pc_G=\left\lfloor\frac{|\delta(\mathcal{P})|}{|\mathcal{P}|-1}\right\rfloor
        \geq \left\lfloor\frac{|E|-|E\setminus\delta(\mathcal{P})|}{f(E)-f(E\setminus \delta(\mathcal{P}))} \right\rfloor
        \geq \min_{A\subseteq E}\left \lfloor \frac{|E\setminus A|}{f(E)-f(A)}\right \rfloor = \lfloor\kappa_{\overrightarrow{\mathbf{1}}}(f)\rfloor.
    \end{aligned}$$

    We now show that $\lfloor\kappa_{\overrightarrow{\mathbf{1}}}(f)\rfloor\geq \pc_G$. Let $A\subseteq E$ be a set of hedges such that $\kappa_{\overrightarrow{\mathbf{1}}}(f)=\frac{|E\setminus A|}{f(E)-f(A)}$. Let $\mathcal{P}$ be the partition of $V$ where each part corresponds to a connected component in the subhedgegraph $(V,A)$. Hence, $|\mathcal{P}|=\Comp(V, A)$. This implies that
    $$\begin{aligned}
        f(E)-f(A) &= (|V|-1) - (|V|-\Comp(V, A)) \ \ \text{(since $G$ is connected)}\\
        &= \Comp(V, A) -1 \\
        &= |\mathcal{P}|-1.
    \end{aligned}$$
    We also note that for every hedge $e\in \delta(\mathcal{P})$, at least one hyperedge in $e$ intersects at least two distinct parts of the partition $\mathcal{P}$, which implies that $e\not\in A$. Hence, $|\delta(\mathcal{P})|\leq |E\setminus A|$. Therefore,
    $$\begin{aligned}
        \lfloor\kappa_{\overrightarrow{\mathbf{1}}}(f)\rfloor=\left \lfloor\frac{|E\setminus A|}{f(E)-f(A)}\right\rfloor
        = \left \lfloor\frac{|E\setminus A|}{|\mathcal{P}|-1}\right \rfloor 
        \geq \left\lfloor\frac{|\delta(\mathcal{P})|}{|\mathcal{P}|-1} \right \rfloor
        \geq \pc_G.
    \end{aligned}$$
\end{proof}

Lemmas~\ref{lemma:strength-in-polynomial-time} and \ref{lemma:partition-connectivity-equals-strength} together imply Theorem~\ref{theorem:partition-connectivity} as follows: Let $G=(V, E)$ be a hedgegraph with $f:2^E\rightarrow \Z_{\ge 0}$ being the associated hedgegraph polymatroid. 
If $G$ is disconnected, then $\pc_G=0$. Suppose $G$ is connected. 
By Lemma~\ref{lemma:partition-connectivity-equals-strength}, we have that $\pc_G = \kappa_{\overrightarrow{\mathbf{1}}}(f)$, where $\overrightarrow{\mathbf{1}}$ is the unit weight function. 
    We have a polynomial time evaluation oracle for the hedgegraph polymatroid $f$. 
    By Lemma~\ref{lemma:strength-in-polynomial-time}, $\kappa_{\overrightarrow{\mathbf{1}}}(f)$ can be computed in polynomial time. 

\section{Hedgegraph Decompositions}\label{section:hedgegraph-decompose}



In this section, we prove Theorems~\ref{theorem:hedgegraph-decomposition},  \ref{theorem:trimming} and \ref{theorem:base-covering} and Corollary \ref{corollary:orientation}. Our proof is by associating a matroid with a hedgegraph. 
Let $G=(V, E)$ be a hedgegraph. Based on the hedgegraph polymatroid $f:2^E\rightarrow \Zg$, we define the function $r:2^E\rightarrow \Zg$ as follows: 
\[
    r(A):=\min\{f(B)+|A \setminus B|: B\subseteq A\}\ \forall\ A\subseteq E. 
\]
It is known that the function $r$ is the rank function of a matroid via standard results on integer-valued polymatroids \cite{edmonds1966submodular}. 
Now, even though we have a matroid based on a hedgegraph, the independent sets of this matroid are not immediately clear. With additional work, we were able to show that the collection of independent sets of this matroid is 
\[
\mathcal{I}_G:=\{F\subseteq E: \text{hedgegraph $(V, F)$ can be trimmed into a forest}\}.
\]

Our proof of the trimming based characterization of independent sets of the matroid with rank function $r$ is rather long. Instead, we take an alternative route to explain the matroid associated with a hedgegraph in this section. 
We first show that $\mathcal{I}_G$ corresponds to the independent sets of a matroid and subsequently, show that $r$ is the rank function of this matroid---see Section \ref{subsection:hedgegraph-matroid}. 
Using these facts, we will prove Theorems~\ref{theorem:hedgegraph-decomposition},  \ref{theorem:trimming} and \ref{theorem:base-covering} and Corollary \ref{corollary:orientation} in Section~\ref{subsection:hedgegraph-decomposition}.


\subsection{Hedgegraph Matroid: Independent Sets and Rank Function}\label{subsection:hedgegraph-matroid}

Let $G=(V, E)$ be a hedgegraph. In this section, we define a matroid associated with a hedgegraph and prove a min-max relation for its rank function. 
For ease of notation, we denote a trimming of a subset $F\subseteq E$ of hedges by a function $Y_F: F\rightarrow \binom{V}{2}$ with the property that $Y_F(e)\subseteq h$ for some hyperedge $h\in e$. We first prove that $(E, \mathcal{I}_G)$ is a matroid.

\begin{lemma}\label{lemma:hedgegraph-matroid-property}
    Let $G=(V, E)$ be a hedgegraph. Then, $\mathcal{M}:=(E, \mathcal{I}_G)$ is a matroid, where: 
    \[
    \mathcal{I}_G:=\{F\subseteq E: \text{hedgegraph $(V, F)$ can be trimmed into a forest}\}.
    \]
\end{lemma}
\begin{proof}
    It suffices to show the two independent set axioms. 
    We first prove the hereditary axiom, i.e., if $B\in \mathcal{I}_G$ and $A\subseteq B$, then $A\in \mathcal{I}_G$. Since $B\in\mathcal{I}_G$, there exists a trimming $Y_B: B\rightarrow \binom{V}{2}$ such that $(V, \cup_{e\in B}Y_B(e))$ is a forest. Since $A\subseteq B$, we can construct a trimming $Y_A: A\rightarrow \binom{V}{2}$ that $Y_A(e)=Y_B(e)$ for every $e\in A$. Hence, $(V, \cup_{e\in A}Y_A(e))$ is a subgraph of $(V, \cup_{e\in B}Y_B(e))$, which is still a forest.

    We now prove the augmentation axiom, i.e., if $A, B\in\mathcal{I}_G$ and $|A|<|B|$, then there exists $e\in B\setminus A$ such that $A+e\in \mathcal{I}_G$. Let $Y_B: B\rightarrow \binom{V}{2}$ be a trimming of hedges in $B$ such that $(V, \cup_{e\in B}Y_B(e))$ is a forest. Since $A\in \mathcal{I}_G$, the hedgegraph $(V, A)$ can be trimmed into a forest. Among all such trimmings, we select a trimming $Y_A: A\rightarrow \binom{V}{2}$ that maximizes
    $$\Phi(Y_A):=|\{e\in A\cap B: Y_A(e)=Y_B(e)\}|.$$
    If there are multiple such trimmings, then we pick an arbitrary one.

    We define graphs $G_{Y_A}:=(V, \cup_{e\in A}Y_A(e))$ and $G_{Y_B}:=(V, \cup_{e\in B}Y_B(e))$. We note that $G_{Y_A}$ and $G_{Y_B}$ are forests. Hence,
    $$\Comp(G_{Y_A})=|V|-|A|>|V|-|B|=\Comp(G_{Y_B}),$$
    where $\Comp(G)$ is the number of connected components in graph $G$. Therefore, there exists an edge $\{u,v\}\in \cup_{e\in B}Y_B(e)$ such that $u$ and $v$ are in different components of $G_{Y_A}$. Let $e'\in B$ be the hedge that is trimmed to $\{u,v\}$, i.e., $Y_B(e')=\{u,v\}$.

    Suppose $e'\in A$. Then, $Y_A(e')\neq \{u,v\}$ since $u$ and $v$ are disconnected in $G_{Y_A}$. We note that the trimming 
    \[
    Y'_A(e):=
    \begin{cases}
        Y_A(e) \text{ if }e\in A\setminus \{e'\}, \\
        \{u,v\} \text{ if }e=e', 
    \end{cases}
    \]
    gives a trimming such that the graph $G_{Y'_A}:=(V, \cup_{e\in A}Y'_A(e))$ is a forest and $\Phi(Y'_A)> \Phi(Y_A)$, a contradiction. 
    Hence, $e'\not\in A$. This implies that $(V, A+e')$ can be trimmed into a forest via the trimming $Y:(A+e')\rightarrow \binom{V}{2}$, where
    $$\begin{aligned}
        Y(e)=
        \begin{cases}
            Y_A(e) & \text{if } e\in A,\\
            \{u,v\} & \text{otherwise.}
        \end{cases}
    \end{aligned}$$
    Therefore, $A+e'\in \mathcal{I}_G$.
\end{proof}

For a hedgegraph $G=(V, E)$, we call $\mathcal{M}:=(E, \mathcal{I}_G)$ as defined in Lemma~\ref{lemma:hedgegraph-matroid-property} to be the \emph{hedgegraph matroid} and let $r: 2^E\rightarrow \Zg$ be the rank function of the hedgegraph matroid. We now obtain a min-max relation for the rank function of a hedgegraph matroid by relating it to the hedgegraph polymatroid.

\begin{lemma}\label{lemma:hedgegraph-matroid-rank-function}
    Let $G=(V, E)$ be a hedgegraph. Let $f:2^E\rightarrow \Zg$ be the hedgegraph polymatroid and let $r:2^E\rightarrow \Zg$ be the rank function of the hedgegraph matroid $\mathcal{M}=(E, \mathcal{I}_G)$. Then, 
    $$r(A)=\min\{f(B)+|A \setminus B|: B\subseteq A\}\ \forall\ A\subseteq E.$$
\end{lemma}
\begin{proof}
    

    
    Let $A\subseteq E$. We show the required equality via matroid intersection. 
    We define
    $$\begin{aligned}
        trims(e)&:=\{(e, \{u, v\}): \{u,v\}\in \binom{V}{2}, \exists\ h\in e \text{ with } u, v\in h\}\ \forall e\in A,\\
        U&:=\bigcup_{e\in A}trims(e), \\
        \mathcal{I}_1&:=\{K\subseteq U: \text{Multigraph}\left(V, \{\{u,v\}: (e, \{u,v\})\in K\}\right)\text{ is a forest}\},
        \ \text{and} \\
        \mathcal{I}_2&:=\left\{K\subseteq U: \left|K\cap trims(e)\right|\leq 1\ \ \forall e\in A\right\}.
    \end{aligned}$$
    In particular, $U$ is the set of all possible trimmings of all hedges in $A$, 
    $\mathcal{I}_1$ is the subset of all possible trimmings of all hedges in $A$ that induce an acyclic subgraph, and  
    $\mathcal{I}_2$ is the subset of trimmings of all hedges in $A$ (since a trimming of a hedge is equivalent to picking exactly one pair of vertices for each hedge). 
    We further define $\mathcal{M}_1:=(U, \mathcal{I}_1)$ and $\mathcal{M}_2:=(U, \mathcal{I}_2)$. We observe that $\mathcal{M}_1$ is a matroid since it is the graphic matroid of the multigraph on vertex set $V$ and edge set $\{\{u,v\}: (e, \{u,v\})\in U\}$. Furthermore, $\mathcal{M}_2$ is also a matroid since it is a partition matroid (with a partition class for each $e\in A$). Let $r_1, r_2:2^U\rightarrow \Zg$ be the matroid rank function of $\mathcal{M}_1, \mathcal{M}_2$, respectively.

    We observe that a subset $F\subseteq A$ admits a trimming $Y_F:F\rightarrow \binom{V}{2}$ to a forest if and only if $\{(e,Y_F(e)): e\in F\}\in \mathcal{I}_1\cap \mathcal{I}_2$. Hence, the maximum size of an independent set of $\mathcal{M}$ contained in $A$ is the maximum size of a common independent set of $\mathcal{M}_1$ and $\mathcal{M}_2$. 
    Thus, by 
    Edmonds' matroid intersection theorem \cite{Edm65-partition}, we know that 
    $$r(A)=\min\{r_1(K)+r_2(U\setminus K):K\subseteq U\}.$$
    
    We first show that $r(A)\leq \min\{f(B)+|A\setminus B|:B\subseteq A\}$. Let $B\subseteq A$ be a set that achieves the minimum, i.e., $f(B) +|A\setminus B|\le f(B') + |A\setminus B'|$ for every $B'\subseteq A$. It suffices to show that $r(A)\leq f(B)+|A\setminus B|$. We define
    \[
        K:=\bigcup_{e\in B}trims(e).
    \]
    We note that $r_1(K)=|V|-\Comp(V,B)=f(B)$ since the maximum size of an acyclic subgraph of a graph with vertex set $V$ and edge set $\{\{u,v\}: (e,\{u,v\})\in K\}$ is  $|V|-\Comp(V, B)$. Meanwhile, $r_2(U\setminus K)=|A\setminus B|$ since $trims(e)\subseteq U\setminus K$ for every $e\in A\setminus B$. Hence,
    $$r(A) \leq r_1(K)+r_2(U\setminus K)=f(B)+|A\setminus B|,$$
    which implies that
    $$r(A)\leq \min\{f(B)+|A\setminus B|: B\subseteq A\}.$$
    
    We now show that $r(A)\geq \min\{f(B)+|A\setminus B|:B\subseteq A\}$. Let $K\subseteq U$ be a set that achieves the minimum, i.e., $r(A)=r_1(K)+r_2(U\setminus K)$. We may assume that for each hedge $e\in A$, either all trimmings of $e$ are in $K$ or none of the trimmings of $e$ are in $K$, i.e., either $trims(e)\subseteq K$ or $trims(e)\cap K=\emptyset$: if $e\in A$ such that $trims(e)\setminus K\neq \emptyset$, then, for $K':=K \setminus trims(e)$, we have that $r_1(K')\le r_1(K)$ and $r_2(U\setminus K') = r_2(U\setminus K)$. We define
    \[
    B:=\left\{e\in A: trims(e)\subseteq K\right\}.
    \]
    Similarly, $r_1(K)=|V|-\Comp(V,B)=f(B)$ since the maximum size of an acyclic subgraph of a graph with vertex set $V$ and edge set $\{\{u,v\}: (e,\{u,v\})\in K\}$ is  $|V|-\Comp(V, B)$. Meanwhile, $r_2(U\setminus K)=|A\setminus B|$ since $trims(e)\subseteq U\setminus K$ for every $e\in A\setminus B$. Hence,
    $$r(A) = r_1(K)+r_2(U\setminus K)=f(B)+|A\setminus B|,$$
    which implies that
    $$r(A)\geq \min\{f(B)+|A\setminus B|: B\subseteq A\}.$$

\end{proof}

\subsection{Hedgegraph Decomposition and Orientation from Hedgegraph Matroid}
\label{subsection:hedgegraph-decomposition}

We restate and prove Theorem~\ref{theorem:trimming}. 
\theoremtrimming*
\begin{proof}
    Let $G=(V,E)$ be a hedgegraph. Let $\mathcal{M}=(E, \mathcal{I}_G)$ be the hedgegraph matroid and $r:2^E\rightarrow \Zg$ be its matroid rank function. We recall that $G$ can be trimmed to a spanning tree if and only if $r(E)=|E|=|V|-1$.

    We first assume that $G$ is $1$-partition connected. Based on Lemma~\ref{lemma:hedgegraph-matroid-rank-function}, we have $r(E)=\min\{f(B)+|E\setminus B|:B\subseteq E\}$. Hence, it suffices to show that $f(B)+|E\setminus B|\ge |V|-1$ for every $B\subseteq E$. Let $B\subseteq E$. 
    Consider the partition $\mathcal{P}$ of $V$ where each part corresponds to a connected component in the sub-hedgegraph $(V, B)$. Hence, $|\mathcal{P}|=\Comp(V,B)$. Since $G$ is $1$-partition connected, we have that $|E\setminus B|=|\delta(\mathcal{P})|\geq |\mathcal{P}|-1$. Therefore,
    $$\begin{aligned}
        f(B)+|E\setminus B| =|V|-\Comp(V,B)+|\delta(\mathcal{P})|\ge |V|-\Comp(V,B)+|\mathcal{P}|-1= |V|-1.
    \end{aligned}$$

    We now assume that $r(E)=|V|-1$ and prove that $G$ is $1$-partition connected. Consider an arbitrary partition $\mathcal{P}$ of $V$. We define $B:=E\setminus \delta(\mathcal{P})$. Therefore, $\Comp(V,B)\geq |\mathcal{P}|$. Since $r(E) = |V|-1$, by Lemma~\ref{lemma:hedgegraph-matroid-rank-function}, we have that $f(B) + |E\setminus B|\ge |V|-1$. Hence, 
    $$\begin{aligned}
        |V|-1 \leq f(B)+|E\setminus B| = |V|-\Comp(V,B)+|E\setminus B|\leq |V|-|\mathcal{P}|+|\delta(\mathcal{P})|,
    \end{aligned}$$
    which implies that $|\delta(\mathcal{P})|\geq |\mathcal{P}|-1$.
\end{proof}

Next, we restate and prove Theorem~\ref{theorem:hedgegraph-decomposition} via Edmonds' base packing theorem 
\cite{Edm65-partition}.
\theoremdecomposition*
\begin{proof}
    Let $G=(V,E)$ be a hedgegraph. Suppose $G$ contains $k$ hedge-disjoint $1$-partition connected sub-hedgegraphs $G_1=(V, E_1), \ldots, G_k=(V, E_k)$. Hence, for every partition $\mathcal{P}$ of $V$, we have
    $$\begin{aligned}
        |\delta_E(\mathcal{P})| = \sum_{i=1}^{k}|\delta_{E_i}(\mathcal{P})| \geq \sum_{i=1}^{k}(|\mathcal{P}|-1)=k(|\mathcal{P}|-1),
    \end{aligned}$$
    which implies that $G$ is $k$-partition connected.

    We now show the converse. Suppose $G$ is $k$-partition connected. 
    Let $\mathcal{M}=(E, \mathcal{I}_G)$ be the hedgegraph matroid and $r:2^E\rightarrow \Zg$ be its matroid rank function. By Edmonds' base packing theorem \cite{Edm65-partition}, the hedgegraph $G$ contains $k$ hedge-disjoint $1$-partition connected sub-hedgegraphs if and only if
    \begin{align}
        k(r(E)-r(A))\leq |E\setminus A|\ \forall A\subseteq E.\label{ineq:edmonds-base-packing}
    \end{align}
    Hence, it suffices to show that (\ref{ineq:edmonds-base-packing}) holds. 
    Since $G$ is $1$-partition connected, we have $r(E)=|V|-1$ based on Theorem~\ref{theorem:trimming}. Also, $f(E)=|V|-1$ because $G$ is connected. Let $A\subseteq E$. By Lemma~\ref{lemma:hedgegraph-matroid-rank-function}, there exists $B\subseteq A$ such that $r(A)=f(B)+|A\setminus B|$. We construct a partition $\mathcal{P}$ of $V$, where each part corresponds to a  connected component in the sub-hedgegraph $(V,B)$. Hence, $|\mathcal{P}|=\Comp(V,B)$. Since $G$ is $k$-partition connected, $|E\setminus B|=|\delta(\mathcal{P})|\geq k(|\mathcal{P}|-1)$. Therefore,
    $$\begin{aligned}
        k(r(E)-r(A)) &= k((|V|-1)-(f(B)+|A\setminus B|)) \ \ \text{(since $r(A)=f(B)+|A\setminus B|)$)}\\
        &= k((|V|-1)-(|V|-\Comp(V,B)+|A\setminus B|)) \ \ \text{(since $f(B)=|V|-\Comp(V,B)$)}\\
        &= k(|\mathcal{P}|-1-|A\setminus B|) \ \ \text{(since $|\mathcal{P}|=\Comp(V,B)$)}\\
        &\leq |E\setminus B|-k|A\setminus B| \ \ \text{(since $|E\setminus B|\geq k(|\mathcal{P}|-1)$)}\\
        &= |E\setminus A|-(k-1)\cdot |A\setminus B|\leq |E\setminus A|.
    \end{aligned}$$
\end{proof}

Next, we restate and prove Theorem~\ref{theorem:base-covering} via Edmonds' matroid covering theorem \cite{Edm65-partition}.
\theoremcovering*
\begin{proof}
    Let $G=(V, E)$ be a hedgegraph. Let $\mathcal{M}=(E, \mathcal{I}_G)$ be the hedgegraph matroid and $r:2^E\rightarrow \Zg$ be its matroid rank function. We note that $E$ is the union of $k$ acyclic-trimmable sets if and only if $E$ is the union of $k$ independent sets. 
    By Edmonds' base covering theorem \cite{Edm65-partition}, $E$ is the union of $k$ independent sets if and only if $k r(A)\ge |A|$ for every subset $A\subseteq E$. We show that $kr(A)\ge |A|$ for every subset $A\subseteq E$ if and only if $|E[\mathcal{P}]|\le k(|V|-|\mathcal{P}|)$ for every partition $\mathcal{P}$ of $V$. 

    Suppose that $kr(A)\ge |A|$ for every subset $A\subseteq E$. Let $\mathcal{P}$ be a partition of $V$. Together with Lemma \ref{lemma:hedgegraph-matroid-rank-function}, for every $B\subseteq A$, we have
    $$|A|\leq kr(A)\leq k(f(B)+|A\setminus B|).$$
    By defining $B:=A$, we have $|A|\leq kf(A)$. We let $A:=E[\mathcal{P}]$. We note that in the subhedgegraph $(V,A)$, the number of connected components is at least $|\mathcal{P}|$, because for every hedge $e\in A$, there is no hyperedge in $e$ intersecting different parts of the partition $\mathcal{P}$. Hence, $f(A) = |V|-\Comp(V,A)\leq |V|-|\mathcal{P}|$, which further implies that
    $$|A|\leq kf(A)\leq k(|V|-|\mathcal{P}|).$$

    We now prove the other direction. Suppose that $|E[\mathcal{P}]|\leq k(|V|-|\mathcal{P}|)$ for every partition $\mathcal{P}$ of $V$. Let $B\subseteq A\subseteq E$. It suffices to show that $k(f(B)+|A\setminus B|)\geq |A|$. Let $\mathcal{P}$ be the partition of $V$ where each part corresponds to a connected component in the subhedgegraph $(V,B)$. Hence, $|\mathcal{P}|=\Comp(V,B)$, which implies that $f(B)=|V|-\Comp(V,B)=|V|-|\mathcal{P}|$. Since $B\subseteq E[\mathcal{P}]$, we have that 
    $$\begin{aligned}
        kf(B)=k(|V|-|\mathcal{P}|)\geq |E[\mathcal{P}]|\geq |B|.
    \end{aligned}$$
    Therefore,
    $$\begin{aligned}
        k(f(B)+|A\setminus B|)\geq |B|+k|A\setminus B|\geq |B|+|A\setminus B|=|A|.
    \end{aligned}$$
\end{proof}

Next, we restate and prove Corollary~\ref{corollary:orientation}. 
\corollaryorientation*
\begin{proof}
    Let $G=(V,E)$ be a hedgegraph and $r\in V$. Suppose $G$ is $k$-partition connected. Then, according to Theorems~\ref{theorem:hedgegraph-decomposition} and \ref{theorem:trimming}, the hedgegraph $G$ contains $k$ hedge-disjoint $1$-partition connected sub-hedgegraphs $G_1=(V, E_1), \ldots, G_k=(V, E_k)$, each of which can be trimmed into a spanning tree. For every spanning tree, we can orient all trimmed edges of the tree away from the root $r$ and lift the same orientation to hedges in $\cup_{i=1}^{k}E_i$, i.e., using the same heads. We orient every other hedge $e\in E\setminus (\cup_{i=1}^{k}E_i)$ arbitrarily. Hence, for every $\emptyset\neq U \subseteq V-r$, we have
    $$\begin{aligned}
        d^{out}_{\overrightarrow{G}}(U)\geq \sum_{i=1}^{k}\left|\left\{e\in E_i: e\in \delta^{out}_{\overrightarrow{G}}(U)\right\}\right|\geq \sum_{i=1}^{k}1=k.
    \end{aligned}$$

    Suppose that $G$ has a rooted $k$-out-hyperarc connected orientation for some choice of the root vertex. Let $\mathcal{P}$ be a partition of the vertex set $V$. For every part $U\subseteq V$ of $\mathcal{P}$ not containing root $r$, we have $d^{out}_{\overrightarrow{G}}(V\setminus U)\geq k$. Hence, $|\delta(P)|\geq \sum_{U\in \mathcal{P}: r\not\in U}d^{out}_{\overrightarrow{G}}(V\setminus U)\ge k\cdot (|\mathcal{P}|-1)$, which implies that $G$ is $k$-partition connected.
\end{proof}
\section{Weak Partition Connectivity of Hedgegraphs}\label{sec:weak-partition-connectivity}
In this section, we relate weak partition connectivity, connectivity, and functional strength of the hedgegraph polymatroid. 
Lemma \ref{lemma:wpc-is-atleast-half-connectivity} shows that weak partition connectivity is at least a half factor of connectivity. Lemma \ref{lemma:hedge-weak-partition-connectivity-polymatroid} shows that weak partition connectivity is at most the functional strength of the hedgegraph polymatroid. 
Lemma \ref{lemma:functional-strength-is-atmost-connectivity} shows that functional strength of the hedgegraph polymatroid is at most connectivity. These results together imply Lemma \ref{lemma:wpc-and-connectivity}. 
We use these results and the log-approximation to functional strength to prove Theorem \ref{theorem:log-approx-connectivity}. 


\begin{lemma}
    \label{lemma:wpc-is-atleast-half-connectivity}
    Let $G=(V, E)$ be a hedgegraph with connectivity $\lambda$ and weak partition connectivity $\wpc_G$. Then, $\lfloor \lambda/2 \rfloor \le \wpc_G$. 
\end{lemma}
\begin{proof}
    Let $\mathcal{P}$ be a partition of $V$ with $\wpc_G=\lfloor\frac{\sum_{e\in E}\left(|\mathcal{P}|-\Comp(\mathcal{P}(e))\right)}{|\mathcal{P}|-1}\rfloor$. 
    Let $e\in E$. We recall that $\mathcal{P}(e)$ is the hedgegraph  obtained from the hedgegraph $(V,\{e\})$ by contracting every part in $\mathcal{P}$ into a single vertex. 
    We define
    $$\mathcal{T}_e:=\{u\in V(\mathcal{P}(e)): u \text{ is not isolated in } \mathcal{P}(e)\}.$$
    We note that $|\mathcal{P}|$ is the number of vertices in $\mathcal{P}(e)$ and $\Comp(\mathcal{P}(e))$ is the number of connected components in $\mathcal{P}(e)$. Hence, we have
    $$|\mathcal{P}|-\Comp(\mathcal{P}(e))\geq \frac{|\mathcal{T}_e|}{2},$$
    since every vertex in $\mathcal{T}_e$ belongs to a connected component of size at least $2$ in the hedgegraph $\mathcal{P}(e)$. This implies that
    $$\begin{aligned}
        \wpc_G = \left \lfloor\frac{\sum_{e\in E}\left(|\mathcal{P}|-\Comp(\mathcal{P}(e))\right)}{|\mathcal{P}|-1}\right \rfloor
        \geq \left \lfloor \frac{1}{2} \sum_{e\in E}\frac{|\mathcal{T}_e|}{|\mathcal{P}|-1}\right \rfloor
        \ge \left \lfloor \frac{1}{2} \sum_{e\in E}\frac{|\mathcal{T}_e|}{|\mathcal{P}|}\right \rfloor.
    \end{aligned}$$

    For every part $S\subseteq V$ of partition $\mathcal{P}$, we note that a hedge $e\in E$ belongs to $\delta(S)$ if and only if the vertex $u(S)$ corresponding to $S$ in $\mathcal{P}(e)$ is not isolated in $\mathcal{P}(e)$, i.e., $u(S)\in \mathcal{T}_e$. Hence, we have
    $$\begin{aligned}
        \sum_{S\in \mathcal{P}}|\delta(S)|=\sum_{S\in \mathcal{P}}\sum_{e\in E}\mathbf{1}_{[e \in \delta(S)]}=\sum_{e\in E}\sum_{S\in \mathcal{P}}\mathbf{1}_{[e \in \delta(S)]}=\sum_{e\in E}\sum_{S\in \mathcal{P}}\mathbf{1}_{[u(S)\in \mathcal{T}_e]}=\sum_{e\in E}|\mathcal{T}_e|.
    \end{aligned}$$
    We also have that $ \sum_{S\in \mathcal{P}}|\delta(S)|\ge \lambda |\mathcal{P}|$ since $|\delta(S)|\ge \lambda$ for every $\emptyset\neq S\subsetneq V$. 
    This implies that
    $$\begin{aligned}
        \wpc_G &\ge \left \lfloor\frac{1}{2} \sum_{e\in E}\frac{|\mathcal{T}_e|}{|\mathcal{P}|} \right \rfloor
        = \left \lfloor\frac{1}{2} \sum_{S\in \mathcal{P}}\frac{|\delta(S)|}{|\mathcal{P}|} \right \rfloor
        \geq \left \lfloor \frac{1}{2} \sum_{S\in \mathcal{P}}\frac{\lambda}{|\mathcal{P}|}\right \rfloor 
        = \left \lfloor\frac{\lambda}{2}\right \rfloor.
    \end{aligned}$$
\end{proof}

We now relate weak partition connectivity to the functional strength of the hedgegraph polymatroid.

\begin{lemma}\label{lemma:hedge-weak-partition-connectivity-polymatroid}
    Let $G=(V,E)$ be a connected hedgegraph, $\wpc_G$ be its weak partition connectivity and $f:2^E\rightarrow \Zg$ be the corresponding hedgegraph polymatroid. Then, $\wpc_G\leq k^*(f)$.
\end{lemma}
\begin{proof}
    Let $A\subseteq E$ be the set of hedges such that $k^*(f)=\lfloor\frac{\sum_{e\in E}(f(A+e)-f(A))}{f(E)-f(A)}\rfloor$. Consider the partition $\mathcal{P}$ of $V$ where each part corresponds to a connected component in the subhedgegraph $(V,A)$. Hence, $|\mathcal{P}|=\Comp(V, A)$. This implies that
    $$\begin{aligned}
        f(E)-f(A) &= (|V|-1) - (|V|-\Comp(V, A)) \ \ \text{(since $G$ is connected)}\\
        &= \Comp(V, A) -1 \\
        &= |\mathcal{P}|-1.
    \end{aligned}$$
    Let $e\in E$. We recall that $\mathcal{P}(e)$ is the hedgegraph  obtained from the hedgegraph $(V,\{e\})$ by contracting every part in $\mathcal{P}$ into a single vertex. Hence, $\Comp(\mathcal{P}(e))=\Comp(V, A+e)$. Therefore, 
    $$\begin{aligned}
        f(A+e)-f(A) &= (|V|-\Comp(V, A+e)) - (|V|-\Comp(V, A)) \\
        &= \Comp(V, A) - \Comp(V, A+e) \\
        &= |\mathcal{P}| - \Comp(\mathcal{P}(e)).
    \end{aligned}$$
    Thus, we have
    $$\begin{aligned}
        k^*(f)=\left\lfloor\frac{\sum_{e\in E}(f(A+e)-f(A))}{f(E)-f(A)}\right\rfloor=\left\lfloor\frac{\sum_{e\in E}(|\mathcal{P}| - \Comp(\mathcal{P}(e)))}{|\mathcal{P}|-1}\right\rfloor\geq \wpc_G.
    \end{aligned}$$

\end{proof}

\begin{remark}
    We have seen that partition connectivity of a hedgegraph is equal to the strength of the hedgegraph polymatroid wrt unit weights in Lemma \ref{lemma:partition-connectivity-equals-strength}. 
    However, weak partition connectivity is not equal to the functional strength of the hedgegraph polymatroid (Lemma \ref{lemma:hedge-weak-partition-connectivity-polymatroid} shows the inequality in one direction). In fact, they are not equal even in hypergraphs. 
    See an example in Appendix~\ref{appendix:couterexample}. Although, there exists a polynomial-time algorithm to compute weak partition connectivity of a hypergraph \cite{FKK03-ori}, we do not know if functional strength of a hypergraph polymatroid can be computed in polynomial time. 
\end{remark}


We show that functional strength of the hedgegraph polymatroid is at most the connectivity of the hedgegraph. 
\begin{lemma}\label{lemma:functional-strength-is-atmost-connectivity}
    Let $G=(V,E)$ be a hedgegraph with connectivity $\lambda>0$ and $f:2^E\rightarrow \Zg$ be the corresponding hedgegraph polymatroid. Then, $k^*(f)\le \lambda$.
\end{lemma}
\begin{proof}
     Let $\emptyset\neq S\subsetneq V$ be a set of vertices such that $\lambda=|\delta(S)|$. Consider the subset $A:=E-\delta(S)$. Let $\mathcal{P}$ be the partition of $V$ with each part corresponding to a connected component in the sub-hedgegraph $(V, A)$. Then, $f(E) = |V|-1$ since $G$ is connected and $f(A)=|V|-|\mathcal{P}|$. 
     For every $e\in E\setminus A=\delta(S)=\delta(\mathcal{P})$, we have 
     \[
     f(A+e) - f(A) 
     = \Comp(V, A) - \Comp(V, A+e)
     = |\mathcal{P}|-\Comp(\mathcal{P}(e)).
     \]
     Hence,
    $$\begin{aligned}
        k^*(f) 
        &\leq \left \lfloor\frac{\sum_{e\in E}\left(f(A+e)-f(A)\right)}{f(E)-f(A)} \right \rfloor\\
        &\leq \left \lfloor\frac{\sum_{e\in E}\left(|\mathcal{P}|-\Comp(\mathcal{P}(e))\right)}{|\mathcal{P}|-1} \right \rfloor\\
        &= \left \lfloor\frac{\sum_{e\in \delta(S)}\left(|\mathcal{P}|-\Comp(\mathcal{P}(e))\right)}{|\mathcal{P}|-1}\right \rfloor \\
        &\leq \left \lfloor \frac{\sum_{e\in \delta(S)}\left(|\mathcal{P}|-1\right)}{|\mathcal{P}|-1}\right \rfloor \ \ \text{(since $\Comp(\mathcal{P}(e))\geq 1$)}\\
        &= |\delta(S)| = \lambda.
    \end{aligned}$$
\end{proof}

Lemmas \ref{lemma:wpc-is-atleast-half-connectivity}, \ref{lemma:hedge-weak-partition-connectivity-polymatroid}, and \ref{lemma:functional-strength-is-atmost-connectivity} together imply Lemma \ref{lemma:wpc-and-connectivity}. 
We now finish the proof of Theorem~\ref{theorem:log-approx-connectivity} using these lemmas and Lemma \ref{lemma:functional-strength-deterministic}.
\theoremlogapproxconnectivity*
\begin{proof}
    Let $G=(V, E)$ be hedgegraph and $f:2^E\rightarrow \mathbb{R}$ be the associated hedgegraph polymatroid. 
     Lemmas \ref{lemma:wpc-is-atleast-half-connectivity}, \ref{lemma:hedge-weak-partition-connectivity-polymatroid}, and \ref{lemma:functional-strength-is-atmost-connectivity} imply that 
    $$\left \lfloor\frac{\lambda}{2}\right \rfloor\le \wpc_G \leq k^*(f) \leq \lambda.$$
    We recall that $f(E)\le |V|-1$. 
    By Lemma~\ref{lemma:functional-strength-deterministic}, there is a deterministic polynomial time algorithm to compute an $O(\log{f(E)})=O(\log{|V|})$-approximation for $k^*(f)$, which implies an $O(\log |V|)$-approximation for $\lambda$ via the above inequalities.
\end{proof}


\section{Sampling and Connectivity}\label{section:sampling-hedge}
In this section, we prove that sampling every hedge with probability $\Omega(\frac{\log n}{\lambda})$ leads to a connected hedgegraph with high probability. 
As mentiond before, Călinescu, Chekuri, Vondrák \cite{cualinescu2009disjoint} showed that 
for a polymatroid $f: 2^{\mathcal{N}}\rightarrow \mathbb{Z}_{\ge 0}$, 
the functional strength $k^*(f)$ is within a $O(\log{f(\mathcal{N})})$-factor of the maximum number of disjoint bases. 
They proved this by showing that 
sampling each element $e\in \mathcal{N}$ independently with probability $p=\Omega(\frac{\log f(\mathcal{N})}{k^*(f)})$ gives a base with constant probability. 
Their sampling result can be strengthened by following their analysis to achieve a success probability of $1-1/f(\mathcal{N})$. We formally state this version of their result below. 
See Appendix~\ref{appendix:sampling-base} for a detailed proof.
\begin{restatable}{lemma}{lemmasamplingbase}\label{lemma:sampling-property}
    Let $f:2^{\mathcal{N}}\rightarrow \mathbb{Z}_{\geq 0}$ be a polymatroid with $f(\mathcal{N})\geq 2$. Let $p:=\min\{1, \frac{10\log f(\mathcal{N})}{k^*(f)}\}$ and $S\subseteq \mathcal{N}$ be a subset obtained by picking each element in $\mathcal{N}$ with probability at least $p$ independently at random. Then, $f(S)=f(\mathcal{N})$ with probability at least $1-\frac{1}{f(\cN)}$. 
\end{restatable}

We now restate and prove Theorem~\ref{theorem:random-sampling}. 
\theoremrandomsampling*
\begin{proof}
    Let $f:2^E\rightarrow \Zg$ be the hedgegraph polymatroid of $G$. We observe that $f(E)=n-1$ since $G$ is connected. Moreover, if a set $A\subseteq E$ of hedges is a base of $f$, then
    $$|V|-\Comp(V, A)=f(A)=f(E)=n-1,$$
    which implies that the sub-hedgegraph $(V, A)$ is connected.
    
    Hence, by Lemma~\ref{lemma:sampling-property}, sampling each hedge in $E$ with probability at least $p:=\min\{1, \frac{10 \log f(E)}{k^*(f)}\}$ gives a connected hedgegraph with probability at least $1-\frac{1}{f(E)}=1-\frac{1}{n-1}$. By Lemmas \ref{lemma:wpc-is-atleast-half-connectivity} and \ref{lemma:hedge-weak-partition-connectivity-polymatroid}, we have
    $$p\leq \frac{10\log f(E)}{k^*(f)}<\frac{10\log n}{k^*(f)}\leq \frac{10\log n}{\wpc_G}\leq \frac{20\log n}{\lambda},$$
    which implies that sampling each hedge $e\in E$ with probability $\Omega(\frac{\log n}{\lambda})$ gives a connected hedgegraph with probability at least $1-\frac{1}{n-1}\ge 1-2/n$.
\end{proof}
\section{Partition Sparsifiers}\label{section:sparsification}
In order to perform partition sparsification in hedgegraphs, we rely on a general result on quotient sparsification of polymatroids due to Quanrud \cite{quanrud2024quotient}. We define the notion of quotients now. 
Let $f:2^\cN\rightarrow \Zg$ be a polymatroid.  For a set $S\subseteq \cN$, the \emph{span} of $S$ is the set of elements with marginal value $0$ with respect to $S$, i.e., $\text{span}(S):=\{e\in \cN:f(e+S)-f(S)=0\}$.
We note that for every two sets $A\subseteq B$, $\text{span}(A)\subseteq \text{span}(B)$ since $f$ is a monotone submodular function. A set $S\subseteq \cN$ is \emph{closed} if $S=\text{span}(S)$.
A set $Q\subseteq \cN$ is a \emph{quotient} of $f$ if $Q=\cN\setminus \text{span}(S)$ for some set $S\subseteq \cN$. We define quotient sparsifiers as follows: let $w: \cN\rightarrow \mathbb{R}_{\geq 0}$ be a weighting of the elements in the ground set. For a subset $S\subseteq \mathcal{N}$, we denote $w(S):=\sum_{e\in S}w(e)$. 
A weighting $\hat{w}:\cN\rightarrow \mathbb{R}_{\geq 0}$ is a \emph{$\varepsilon$-quotient sparsifier} of $f$ with respect to $w$ if $|w(Q)-\hat{w}(Q)|\leq \varepsilon\cdot w(Q)$ for every quotient $Q\subseteq \cN$ of $f$. 
Quanrud \cite{quanrud2024quotient} showed the following quotient sparsification result for polymatroids.

\begin{theorem}\cite{quanrud2024quotient}\label{theorem:improve-quotient-sparsification}
    Let $f:2^{\cN}\rightarrow \Zg$ be a polymatroid with non-negative weights $w:\cN\rightarrow \mathbb{R}_{\geq 0}$, $r=f(\cN)$, and $\varepsilon>0$. Suppose for every integer $t\in \Z_+$, there are $r^{O(t)}$ quotients $Q$ of $f$ with $w(Q) \le t\kappa_f$. Then, there exists a $\varepsilon$-quotient sparsifier $\hat{w}:\cN\rightarrow \mathbb{R}_{\geq 0}$ of $f$ with respect to $w$ with support size $|\text{support}(\hat{w})|=O(r\log r / \varepsilon^2)$ and moreover, given evaluation oracle access to $f$ and the weighting $w$, such a sparsifier can be computed with high probability in randomized polynomial time.
\end{theorem}

We begin by showing that quotients of the hedgegraph polymatroid correspond to a subset of hedges crossing some vertex partition.

\begin{lemma}\label{lemma:hedgegraph-cut-quotient}
    Let $G=(V,E)$ be a hedgegraph and $f:2^E\rightarrow \Zg$ be the corresponding hedgegraph polymatroid. Then, a subset $Q\subseteq E$ is a quotient of $f$ if and only if $Q=\delta(\mathcal{P})$ for some partition $\mathcal{P}$ of $V$.
\end{lemma}
\begin{proof}
    We first show that for every partition $\mathcal{P}$ of $V$, the hedge set $Q=\delta(\mathcal{P})$ is a quotient of $f$. We note that for every hedge $e\in \delta(\mathcal{P})$, we have
    $$\Comp(V, (E\setminus Q)+e) < \Comp(V, E\setminus Q).$$
    This implies that
    $$\begin{aligned}
        f((E\setminus Q)+e) &= |V| - \Comp(V, (E\setminus Q)+e) \\
        & > |V| - \Comp(V, E\setminus Q) \\
        &= f(E\setminus Q),
    \end{aligned}$$
    which shows that $E\setminus Q$ is a closed set and $Q$ is a quotient of $f$.

    We now prove the other direction. Let $Q\subseteq E$ be a quotient of $f$. We construct a partition $\mathcal{P}$ of $V$, where each part corresponds to one connected component in the subhedgegraph $(V, E\setminus Q)$. We note that for every hedge $e\in E\setminus Q$, we have $e\not\in \delta(\mathcal{P})$, which implies that $\delta(\mathcal{P})\subseteq Q$. Meanwhile, $E\setminus Q$ is a closed set, which implies that for every hedge $e\in Q$, we have $f((E\setminus Q)+e)>f(E\setminus Q)$. This shows that
    $$\Comp(V, (E\setminus Q)+e) < \Comp(V, E\setminus Q)=|\mathcal{P}|,$$
    which further implies that $e\in \delta(\mathcal{P})$. Hence, $Q\subseteq \delta(\mathcal{P})$, together with $\delta(\mathcal{P})\subseteq Q$ indicating that $Q=\delta(\mathcal{P})$.
\end{proof}


Next, we show that the hedgegraph polymatroid satisfies the condition in Theorem~\ref{theorem:improve-quotient-sparsification}.

\begin{lemma}\label{lemma:hedge-counting}
    Let $G=(V, E)$ be a $n$-vertex hedgegraph with hedge-weights $w: E \rightarrow \Rg$. Let $f:2^E\rightarrow \Zg$ be the corresponding hedgegraph polymatroid. Then, for every integer $t\geq 1$, there are at most $n^{O(t)}$ quotients $Q$ of $f$ with $w(Q)\le t\kappa_w(f)$.
\end{lemma}
\begin{proof}
    We first recall the ideas of the proof in Quanrud \cite{quanrud2024quotient}. For a polymatroid $f:2^{\mathcal{N}}\rightarrow \Zg$ with weights $w: \cN\rightarrow \Rg$, we consider a random collection of quotients of $f$ generated as follows. Initially, let $R=\emptyset$. While $f(R)<f(\mathcal{N})-t-1$, 
    we sample a random element $e\in \mathcal{N}\setminus R$ with probability proportional to its weight $w(e)$, add the element $e$ to $R$, and replace $R$ with its span $\text{span}(R)$ to ensure that it is always a closed set. 
    At the end of the while-loop, we obtain a random closed set $R$ with $f(R)\geq f(\mathcal{N})-t-1$. We return all quotients disjoint from $R$ with total weight at most $t\kappa_w(f)$. Quanrud \cite{quanrud2024quotient} showed that there are at most $|\mathcal{N}|
    ^{t+1}$ quotients disjoint from $R$ and for every fixed quotient $Q\subseteq \mathcal{N}$ with weight at most $t\kappa_w(f)$, $Q$ is disjoint from $R$ with probability at least $1/\binom{f(\mathcal{N})}{t}$, resulting in an overall bound of $O(|\mathcal{N}|^{t+1} f(\mathcal{N})^{t})$ quotients of weight at most $t\kappa_w(f)$.
    
    We improve this bound for the hedgegraph polymatroid by showing that there are at most $(t+2)^{O(t)}$ quotients disjoint from $R$ in the following paragraph. 
    By Quanrud's argument, we already know that every fixed quotient with weight at most $t\kappa_w(f)$ is 
    disjoint from $R$ with probability at least $1/\binom{f(E)}{t}=n^{-O(t)}$. This implies that there are at most $n^{O(t)}\cdot (t+2)^{O(t)}=n^{O(t)}$ quotients of $f$ with weight at most $t\kappa_w(f)$.

    We now consider the hedgegraph polymatroid $f:2^E\rightarrow \Zg$ and prove that there are at most $(t+2)^{O(t)}$ quotients disjoint from $R$. We note that $R$ is a closed set of $f$ with
    $$f(R)\geq f(E)-t-1=|V|-t-2,$$
    which implies that $\Comp(V, R)=|V|-f(R) \leq t+2$. Since $R$ is a closed set, $E\setminus R$ is a quotient. We construct a partition $\mathcal{P}$ of $V$, where each part corresponds to one connected component in the subhedgegraph $(V,R)$. For every hedge $e\in R$, we have $e\not\in \delta(\mathcal{P})$, which implies that $\delta(\mathcal{P})\subseteq E\setminus R$. Meanwhile, for every hedge $e\in R$, we have $f(R+e)>f(R)$ because $R$ is a closed set. This shows that
    $$\Comp(V, R\cup \{e\})<\Comp(V,R)=|\mathcal{P}|,$$
    which further implies that $e\in \delta(\mathcal{P})$. Hence, $E\setminus R\subseteq \delta(\mathcal{P})$, together with $\delta(\mathcal{P})\subseteq E\setminus R$ indicating that $E\setminus R=\delta(\mathcal{P})$. By Lemma~\ref{lemma:hedgegraph-cut-quotient}, for every quotient $Q\subseteq E\setminus R$, we have $Q=\delta(\mathcal{P}')$ for some partition $\mathcal{P}'$ of $V$. We note that $\mathcal{P}'$ can be obtained by merging some parts in $\mathcal{P}$. In each merge operation, we can select two different parts and merge them. Since $|\mathcal{P}|\leq t+2$, we can do at most $t+1$ merge operations, of which each has at most $\binom{|\mathcal{P}|}{2}\leq (t+2)^2$ options.
    Hence, there are at most
    $$\sum_{i=0}^{t+1}\binom{(t+2)^2}{i}\leq (t+2)^{2(t+1)}$$
    quotients disjoint from $R$. 
    
\end{proof}

We observe that for a hedgegraph polymatroid $f$ associated with a hedgegraph $G=(V, E)$, we have that $r=f(E)\le |V|$. 
Theorem \ref{theorem:improve-quotient-sparsification} and Lemmas \ref{lemma:hedgegraph-cut-quotient} and \ref{lemma:hedge-counting}, along with the observation that the restriction of a hedgegraph polymatroid to a subset of hedges is again a hedgegraph polymatroid, together 
imply Theorem~\ref{theorem:sparsification}.
\section{Conclusion}\label{sec:conclusion}
We investigated structural and algorithmic aspects of hedgegraphs via the hedgegraph polymatroid. The hedgegraph polymatroid allowed us to use properties of matroids and
submodularity even though the cut function of a hedgegraph is not submodular. We illustrated some of the consequences of this investigation.
We mention one more polymatroidal result that leads to an interesting conclusion on hedgegraphs: there exists 
an $O(\log^2{|V|})$-competitive randomized algorithm to compute an online packing of maximum number of disjoint connected spanning hedgegraphs when vertices are known apriori and 
hedges arrive online. This is a direct consequence of applying the polymatroidal result of \cite{CCZ25} to the hedgegraph polymatroid. 
We believe that the hedgegraph polymatroid is a promising tool for both algorithmic and structural aspects of hedgegraphs. We mention two interesting open questions: (1) Does there exists a deterministic approximation scheme for the connectivity of a hedgegraph? Even a constant factor deterministic approximation for hedgegraph connectivity is open. This leads to the second open question: (2) Is weak partition connectivity of a hedgegraph computable in polynomial time? 

\bibliographystyle{abbrv}
\bibliography{references}

\newpage
\appendix
\section{Submodularity of Hedgegraph Polymatroid}\label{appendix:hedgegraph-polymatroid}
In this section, we prove that the hedgegraph polymatroid function is submodular.

\begin{lemma}
    Let $G=(V,E)$ be a hedgegraph. We define the hedgegraph polymatroid function $f:2^E\rightarrow \Zg$ as 
    \[
        f(A) := |V| - \Comp(V, A)\ \forall\ A\subseteq E,
    \]
    where $\Comp(V, A)$ is the number of connected components in the subhedgegraph $(V,A)$. Then, $f$ is a submodular function.
\end{lemma}
\begin{proof}
    Let $G_h:=(V,H)$ be the underlying hypergraph of $G$, i.e., $H$ is the multi-set $\{f\subseteq V: f\in e \text{ for some } e\in E\}$. Let $h:2^V\rightarrow \Zg$ be the hypergraph polymatroid of $G_h$ such that
    \[
        h(B) = |V| - \Comp(V,B)\ \forall\ B\subseteq H,
    \]
    where $\Comp(V, B)$ is the number of connected components in the subhypergraph $(V,B)$. It is well-known that $h$ is a submodular function (e.g., see \cite{cualinescu2009disjoint}). We note that for every hedge set $A\subseteq E$, we have
    $$\begin{aligned}
        f(A) = |V| - \Comp(V,A) = |V| - \Comp\left(V,\bigcup_{e\in A}e\right) = h\left(\bigcup_{e\in A}e\right).
    \end{aligned}$$
    Hence, for two arbitrary hedge sets $A_1, A_2 \subseteq E$, we have
    $$\begin{aligned}
        f(A_1)+f(A_2) &= h\left(\bigcup_{e\in A_1}e\right) + h\left(\bigcup_{e\in A_2}e\right) \\
        &\geq h\left(\bigcup_{e\in A_1\cap A_2}e\right) + h\left(\bigcup_{e\in A_1\cup A_2}e\right) \ \ \text{(by the submodularity of $h$)}\\
        &= f(A_1\cap A_2) + f(A_1\cup A_2),
    \end{aligned}$$
    which implies the subodularity of $f$.
\end{proof}
\section{Hypergraphs with $\wpc_G<k^*(f)$}\label{appendix:couterexample}

In this section, we exhibit a hypergraph $G=(V,E)$ whose weak partition connectivity is strictly less than the functional strength of the corresponding hypergraph polymatroid. 

\begin{lemma}
    There exists a hypergraph $G=(V, E)$ with hypergraph polymatroid $f:2^E\rightarrow \Z_{\ge 0}$ such that $\wpc(G)<k^*(f)$. 
\end{lemma}
\begin{proof}
Consider the hypergraph $G=(V,E)$ with $4$ vertices and $5$ hyperedges, where $V:=\{A,B,C,D\}$ and $E:=\{e_i:i\in [5]\}$ with $e_1=\{A,B\}$, $e_2=e_3=\{A,C,D\}$, and $e_4=e_5=\{B,C,D\}$. The hypergraph $G$ is shown in Figure~\ref{figure:counterexample}.

\begin{figure}[h]
    \centering
    \includegraphics[width=0.4\linewidth]{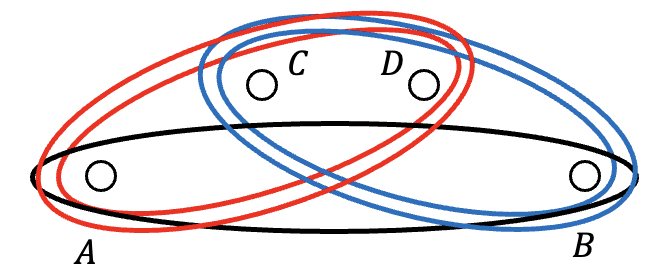}
    \caption{A figure of hypergraph $G$, where $e_1=\{A,B\}$ is the black hyperedge, $e_2$ and $e_3$ are the two red parallel hyperedges, and $e_4$ and $e_5$ are the two blue parallel hyperedges.}
    \label{figure:counterexample}
\end{figure}

We first show an upper bound on the weak partition connectivity of $G$. Consider the partition $\mathcal{P}:=\{\{A\},\{B\},\{C,D\}\}$. Then, $|\mathcal{P}|=3$, $\Comp(\mathcal{P}(e_i))=2$ for every $i\in [5]$. Hence,
$$\wpc_G\leq \left\lfloor\frac{\sum_{e\in E}\left(|\mathcal{P}|-\Comp(\mathcal{P}(e))\right)}{|\mathcal{P}|-1}\right\rfloor=\left\lfloor \frac{5}{2}\right\rfloor=2.$$

We now compute the functional strength of the corresponding hypergraph polymatroid. We will show that $k^*(f)= 3$. 
We recall that $f(S)=|V|-\Comp(V,S)$ for every $S\subseteq E$. Hence, $f(\emptyset)=0$, $f(\{e_1\})=1$, $f(\{e_2\})=f(\{e_3\})=f(\{e_2,e_3\})=2$, $f(\{e_4\})=f(\{e_5\})=f(\{e_4, e_5\})=2$, and all other subsets $S\subseteq E$ have $f(S)=3$. We recall that $$k^*(f)=\min_{A\subseteq E} \left\lfloor \frac{\sum_{e\in E}\left(f(A+e)-f(A)\right)}{f(E)-f(A)}\right\rfloor.$$
Thus, we need to consider all subsets $A\subseteq E$ with $f(A)<f(E)=3$.

For the case of $A=\emptyset$, we have 
$$\left\lfloor \frac{\sum_{e\in E}\left(f(A+e)-f(A)\right)}{f(E)-f(A)}\right\rfloor=\left\lfloor\frac{9}{3}\right\rfloor=3.$$

For the case of $A=\{e_1\}$, we have 
$$\left\lfloor \frac{\sum_{e\in E}\left(f(A+e)-f(A)\right)}{f(E)-f(A)}\right\rfloor=\left\lfloor\frac{8}{2}\right\rfloor=4.$$

For every other subset $A\subseteq E$ with $f(A)=2$, we have 
$$\left\lfloor \frac{\sum_{e\in E}\left(f(A+e)-f(A)\right)}{f(E)-f(A)}\right\rfloor=\left\lfloor\frac{3}{1}\right\rfloor=3.$$
\end{proof}
\section{Sampling Elements Gives a Base}\label{appendix:sampling-base}
In this section, we give a detailed proof of Lemma~\ref{lemma:sampling-property}. We restate it as follows.

\lemmasamplingbase*
\begin{proof}
    For ease of expression, we define $r:=f(\mathcal{N})\geq 2$. If $k^*(f)\leq 10\log r$, then we have $p=1$, which implies that $S=\mathcal{N}$ which is a base of $f$. Henceforth, we may assume that $k^*(f)>10\log r$ and consequently, $p=\frac{10\log r}{k^*(f)}$. Let $\sigma=(e_1, e_2, \ldots, e_{|\mathcal{N}|})$ be a uniformly random permutation of the elements in $\mathcal{N}$. For every $i\in [0,|\mathcal{N}|]$, we define $\mathcal{N}_i:=\{e_j:j\in [i]\}$ and $S_i:=S\cap \mathcal{N}_i$. We note that $S_0=\mathcal{N}_0=\emptyset$ and $f_{S_0}(\mathcal{N})=r$. For a subset $A\subseteq \mathcal{N}$, we define the function $f_A:2^\mathcal{N}\rightarrow \mathbb{Z}$ as $f_A(S):=f(A\cup S)-f(A)$ for every $S\subseteq \mathcal{N}$. We note that $f_A$ is submodular for every $A\subseteq \mathcal{N}$.
    
    Let $i\in [|\mathcal{N}|]$. We consider the distribution of element $e_i$ conditioned on $S_{i-1}$. Since $\sigma$ is a uniformly random permutation, $e_i$ can be an arbitrary element in $\mathcal{N}\setminus S_{i-1}$, of which each has the same probability. That is, for every element $e\in \mathcal{N}$,
    $$\begin{aligned}
        \mathbf{Pr}_{\sigma, S}[e_i=e|S_{i-1}] = 
        \begin{cases}
            0, \ \ & \text{ if }e\in S_{i-1} \\
            \frac{1}{|\mathcal{N}|-|S_{i-1}|} \ \ & \text{ if }e\in \mathcal{N}\setminus S_{i-1}.
        \end{cases}
    \end{aligned}$$
    We note that element $e_i$ is picked with probability at least $p$. If element $e_i$ is picked, then $S_i=S_{i-1}+e_i$. Otherwise, $S_i=S_{i-1}$. Hence, we have
    $$\begin{aligned}
        \mathbb{E}_{\sigma, S}[f_{S_{i-1}}(\mathcal{N})-f_{S_i}(\mathcal{N}) | S_{i-1}] &= \sum_{e\in \mathcal{N}\setminus S_{i-1}} \frac{1}{|\mathcal{N}|-|S_{i-1}|} \cdot \mathbb{E}_{\sigma, S}[f_{S_{i-1}}(\mathcal{N})-f_{S_i}(\mathcal{N}) | S_{i-1} \ \text{and} \ e_i=e] \\
        &\geq \sum_{e\in \mathcal{N}\setminus S_{i-1}} \frac{p}{|\mathcal{N}|-|S_{i-1}|} \cdot f_{S_{i-1}}(e) \\
        &\geq \frac{p}{|\mathcal{N}|}\cdot \sum_{e\in \mathcal{N}\setminus S_{i-1}}f_{S_{i-1}}(e) \\
        &\geq \frac{p}{|\mathcal{N}|}\cdot k^*(f)\cdot f_{S_{i-1}}(\mathcal{N}),
    \end{aligned}$$
    where the last inequality is by the definition of $k^*(f)$. This implies that
    $$\begin{aligned}
        \mathbb{E}_{\sigma, S}[f_{S_{i}}(\mathcal{N})] &= \mathbb{E}_{\sigma, S}[f_{S_{i-1}}(\mathcal{N})-\mathbb{E}_{\sigma, S}[f_{\mathcal{S}_{i-1}}(\mathcal{N})-f_{\mathcal{S}_i}(\mathcal{N})| S_{i-1}]] \\
        &\leq \mathbb{E}_{\sigma, S}\left[\left(1-\frac{p}{|\mathcal{N}|}\cdot k^*(f)\right)\cdot f_{S_{i-1}}(\mathcal{N})\right]\\
        &= \left(1-\frac{p}{|\mathcal{N}|}\cdot k^*(f)\right)\cdot \mathbb{E}_{\sigma, S}[f_{S_{i-1}}(\mathcal{N})].
    \end{aligned}$$

    By setting $i=|\mathcal{N}|$, we have
    $$\begin{aligned}
        \mathbb{E}_{\sigma, S}[f_{S_{|\mathcal{N}|}}(\mathcal{N})] &\leq \left(1-\frac{p}{|\mathcal{N}|}\cdot k^*(f)\right)^{|\mathcal{N}|} \cdot \mathbb{E}_{\sigma, S}[f_{S_0}(\mathcal{N})] \\
        &= \left(1-\frac{p}{|\mathcal{N}|}\cdot k^*(f)\right)^{|\mathcal{N}|} \cdot r \\
        &< \exp(-p\cdot k^*(f))\cdot r = \exp(-10\log r)\cdot r < \frac{1}{r}.
    \end{aligned}$$
    According to Markov's inequality, we have
    $$\begin{aligned}
        \mathbf{Pr}_{\sigma, S}[f_{S_{|\mathcal{N}|}}(\mathcal{N}) \geq 1] \leq \mathbb{E}_{\sigma, S}[f_{S_{|\mathcal{N}|}}(\mathcal{N})] < \frac{1}{r},
    \end{aligned}$$
    which shows that
    $$\begin{aligned}
        \mathbf{Pr}_{\sigma, S}[f(S_{|\cN|})=f(\cN)] = \mathbf{Pr}_{\sigma, S}[f_{S_{|\mathcal{N}|}}(\mathcal{N}) = 0] > 1-\frac{1}{r}.
    \end{aligned}$$
    This implies that the probability of $S=S_{|\cN|}$ being a base is greater than $1-\frac{1}{r}$.
    
\end{proof}

\end{document}